\documentclass[11pt]{article}
\usepackage{amsmath, amssymb, amsthm}
\usepackage{graphicx,psfrag}
\usepackage{mathtools}
\usepackage{bm}
\usepackage{bbm}
\usepackage{enumitem}
\usepackage{subfigure}
\usepackage{natbib}
\usepackage[mathscr]{euscript}
\usepackage{url} 
\usepackage[bottom]{footmisc}
\usepackage{booktabs}
\usepackage{multirow}
\usepackage{color}
\usepackage[hypertexnames=false]{hyperref}
\hypersetup{
colorlinks = true,
allcolors = blue
} 
\usepackage{caption} 
\usepackage{adjustbox} 
\mathtoolsset{showonlyrefs}

\usepackage[a4paper,                   
    left=1.1in,right=1.1in,top=1in,bottom=1.2in,%
            footskip=.7in]{geometry}
    
\usepackage[toc,page,header]{appendix}
\usepackage{minitoc}   
\usepackage{xcolor}
\usepackage[linesnumbered,ruled,vlined]{algorithm2e}
\usepackage{wrapfig}

\makeatletter
\newcommand{\algorithmfootnote}[2][\footnotesize]{%
  \let\old@algocf@finish\@algocf@finish
  \def\@algocf@finish{\old@algocf@finish
    \leavevmode\rlap{\begin{minipage}{\linewidth}
    #1#2
    \end{minipage}}%
  }%
}
\makeatother

\DeclareMathOperator*{\argmax}{arg\,max}

\def\cN{\mathcal{N}}
\def\cX{\mathcal{X}}
\def\cV{\mathcal{V}}
\def\cE{\mathcal{E}}

\def\bbS{\mathbb{S}}
\def\bbR{\mathbb{R}}

\def\bbN{\mathbb{N}}
\def\ind{\mathbbm{1}}   

\def\pimin{\pi_{\mathrm{min}}}
\def\tmix{\tau}

\def\N{M}
\def\R{R}
 
\def\T{\top} 
\def\bK{\mathsf{K}}
\def\bP{\mathsf{P}}
\def\bX{\bm{X}}
\def\bI{\bm{I}}
\def\by{\bm{y}}
\def\bz{\bm{z}}

\def\bbeta{\bm{\beta}}
\def\bA{\mathsf{A}} 
\def\PR{\mathsf{P}_0}
 
\def\lbda{\alpha} 
\def\mvn{\mathrm{MVN}}
\def\d{\mathrm{d}}

\def\barE{\bar{E}}

\def\din{d_{\mathrm{in}}}              
\def\dout{d_{\mathrm{out}}}            

\def\Pa{\mathrm{Pa}} 
\def\Ch{\mathrm{Ch}}  
\newcommand{\lmin}{\lambda_{\mathrm{min}}}
\newcommand{\lmax}{\lambda_{\mathrm{max}}}

\def\vmin{\underline{\nu}}
\def\vmax{\overline{\nu}}

\theoremstyle{remark}
\newtheorem{remark}{Remark}

\newtheorem{proposition}{Proposition}
\newtheorem{theorem}{Theorem}
\newtheorem{corollary}{Corollary}
\newtheorem{lemma}{Lemma}
\newtheorem{example}{Example}
\newtheorem{definition}{Definition}

\title{Dimension-free Relaxation Times of Informed MCMC Samplers on Discrete Spaces} 
\usepackage{authblk} 
\author{Hyunwoong Chang$^1$ and  Quan Zhou$^{2,}$\thanks{Corresponding author: quan@stat.tamu.edu}}
\date{}
\affil{$^1$Department of Mathematical Sciences, The University of Texas at Dallas\\
$^2$Department of Statistics, Texas A\&M University}

\begin{document}
\maketitle

\begin{abstract}
Convergence analysis of Markov chain Monte Carlo methods in high-dimensional statistical applications is increasingly recognized. In this paper, we develop general mixing time bounds for Metropolis-Hastings algorithms on discrete spaces by building upon and refining some recent theoretical advancements in Bayesian model selection problems. 
We establish sufficient conditions for a class of informed Metropolis-Hastings algorithms to attain relaxation times that are independent of the problem dimension. 
These conditions are grounded in the high-dimensional statistical theory and allow for possibly multimodal posterior distributions. 
We obtain our results through two independent techniques: the multicommodity flow method and single-element drift condition analysis; we find that the latter yields a  slightly  tighter mixing time bound. 
Our results  are readily applicable to a broad spectrum of statistical problems with discrete parameter spaces,  as we demonstrate using both theoretical and numerical examples. 
\end{abstract}
\noindent
\textit{Keywords:}
{
\small
Drift condition; Finite Markov chains; Informed Metropolis-Hastings; Mixing time; Model selection; Multicommodity flow; Random walk Metropolis-Hastings; Restricted spectral gap
}

\section{Introduction}\label{sec:intro}

\subsection{Convergence of MCMC algorithms} \label{sec:intro-mcmc}
Approximating probability distributions with unknown normalizing constants is a ubiquitous challenge in data science. 
In Bayesian statistics, the posterior probability distribution can be obtained through the Bayes rule, but without conjugacy, calculating its normalizing constant is a formidable challenge.   
In such scenarios, Markov chain Monte Carlo (MCMC) methods are commonly employed to simulate a Markov chain whose stationary distribution coincides with the target distribution. The convergence assessment of MCMC methods usually relies on an empirical diagnosis of Markov chain samples, such as effective sample size~\citep{robert1999monte, gong2016practical} and Gelman-Rubin statistic~\citep{gelman1992inference}, 
but the results can sometimes be misleading due to the intrinsic difficulty in detecting the non-convergence~\citep{cowles1996markov}.   
Meanwhile, theoretical analyses of the samplers’ convergence rates offer a complementary perspective: they not only illuminate how MCMC performance scales with the problem dimensionality but also guide the selection of suitable algorithms for complex or high-dimensional settings.  
A widely used metric for assessing the convergence is the mixing time, which indicates the number of iterations required for the MCMC sampler to get sufficiently close to the target posterior distribution in  total variation distance.   
Most existing studies on the mixing of MCMC algorithms consider Euclidean spaces. 
In particular, a very rich theory for Metropolis-Hastings (MH) algorithms with log-concave target distributions has been developed~\citep{Dwivedi2019log, durmus2019high, cheng2018underdamped, dalalyan2017theoretical, chewi2021optimal}, providing useful insights into the sampling complexity and guidance on the tuning of algorithm parameters. 
These results encompass a wide range of statistical models  owing to the Bernstein-von Mises theorem~\citep{belloni2009computational, tang2022computational}, which suggests that the posterior distribution becomes approximately normal (and thus log-concave) as the sample size tends to infinity.  
For Gibbs sampling on continuous spaces, model-specific complexity bounds have been obtained for various Bayesian problems, including probit regression~\citep{qin2019convergence, qin2022wasserstein}, linear mixed models~\citep{jin2022dimension, yang2023complexity}, and general hierarchical models~\citep{ascolani2024dimension}.    
In contrast, the literature on the complexity of MCMC samplers for discrete statistical models is limited. Further, since every discrete space comes with its own combinatorial structure, researchers often choose to investigate each problem on a case-by-case basis rather than formulating a unified theory for arbitrary discrete spaces.  
Notable examples  include variable selection~\citep{yang2016computational},  community detection~\citep{zhuo2021mixing}, structure learning~\citep{zhou2021complexity}, 
and classification and regression trees (CART)~\citep{kim2023mixing}.  

In an earlier study from the authors~\citep{zhou2021complexity}, we proposed a general theoretical framework for studying the complexity of MCMC sampling  for high-dimensional statistical models with discrete parameter spaces. By assuming a mild unimodal condition on the posterior distribution, 
we derived mixing time bounds for random walk MH algorithms that were sharper than existing ones.  
Here ``unimodal'' means that other than the global maximizer, every state has a ``neighbor'' with strictly larger posterior probability, where two states are said to be neighbors if the proposal probability from one to the other is nonzero. 
Selecting this neighborhood relation often entails a trade-off, and the result of \citet{zhou2021complexity} shows that rapid mixing can be achieved if (i) the neighborhood is large enough so that the posterior distribution becomes unimodal, and (ii) the neighborhood is not too large so that the proposal probability of each neighboring state is not exceedingly small. 
Techniques from high-dimensional statistical theory can be used to establish the unimodality in a way similar to how posterior consistency is proved, but the analysis is often highly complicated and tailored to individual problems. 
We will not explore these statistical techniques in depth in this paper; instead, our focus will be the mixing time analysis under the unimodal condition or  a more general condition allowing for mulitmodality  that will be introduced later. 

\subsection{Main contributions of this work}\label{sec:intro-contributions}
Our first objective is to extend the result of~\citet{zhou2021complexity} to more sophisticated MH schemes that do not use random walk proposals. We consider the informed MH algorithm proposed in~\citet{zanella2020informed},  which emulates the behavior of gradient-based MCMC samplers on Euclidean spaces~\citep{duane1987hybrid, roberts2002langevin, girolami2011riemann} and has gained increasing attention among the MCMC practitioners~\citep{grathwohl2021oops, zhang2022langevin}.  
The proposal schemes used in the informed algorithms always assess the local posterior landscape surrounding the current state and then tune the proposal probabilities to prevent the sampler from visiting states with low posterior probabilities. 
However, such informed schemes do not necessarily result in faster mixing, because the acceptance probabilities can be extremely low for proposed states, and it was shown in~\citet{zhou2022dimension} that, for Bayesian variable selection, naive informed MH algorithms can even mix much more slowly than random walk MH algorithms. 
We will show that this acceptance probability issue can be overcome by generalizing the ``thresholding'' idea introduced in~\cite{zhou2022dimension}.  
Moreover, if the posterior distribution is unimodal and tails decay sufficiently fast, we prove that there always exists an informed MH scheme whose relaxation time can be bounded by a  constant, independent of the problem dimension, which immediately leads to a nearly optimal bound on the mixing time.
Our result significantly generalizes the finding of~\cite{zhou2022dimension}, which only considered the high-dimensional variable selection problem.  

Delving into the more technical aspects, we investigate and compare two different approaches to obtaining  sharp mixing time bounds on discrete spaces: the multicommodity flow method~\citep{gine1996lectures, sinclair1992improved}, and the single-element drift condition~\citep{jerison2016drift}. 
The former bounds the spectral gap of the transition matrix by identifying likely paths connecting any two distinct states. Compared to another path argument, known as ``canonical path ensemble'' and commonly used in the statistical literature~\citep{yang2016computational, zhuo2021mixing, kim2023mixing, chang2022rapidly}, the multicommodity flow method is more flexible and can yield tighter bounds.  
The drift condition method is based on a coupling argument and stands as the most popular technique for deriving the convergence rates of MCMC algorithms on general state spaces~\citep{rosenthal1995minorization, roy2007convergence, fort2003geometric, johndrow2020scalable}. 
A notable difference from the path method is that this approach directly bounds the total variation distance from the stationary distribution without analyzing the spectral gap. 
The existing literature suggests that both methods have their own unique strengths, and which method yields a better mixing time bound depends on the problem~\citep{jerrum1996markov, guruswami2000rapidly, anil2001coupling}.   
To our knowledge, the two approaches have never been compared for analyzing MCMC algorithms for discrete-space statistical models. Indeed, the only works we are aware of that use drift conditions in these contexts are \citet{zhou2022dimension} and \citet{kim2023mixing}, which considered variable selection and CART, respectively. 
We will demonstrate how  the two methods can be applied under the general framework considered in this paper, and it will be shown that the drift condition approach  yields a slightly sharper bound.  

As the last major contribution of this work, we further extend our general theory beyond unimodal settings. 
For multimodal target distributions, while various mixing time bounds have been obtained~\citep{guan2007small, woodard2009conditions, zhou2022rapid}, it is generally impossible to obtain rapid mixing results where the mixing time grows only polynomially with respect to the problem dimension. 
This is because the definition of mixing time considers the worst-case scenario regarding the choice of the initial distribution, and the chain can easily get stuck if it is initialized at a local mode. When these local modes (other than the global one) possess only negligible posterior mass and are unlikely to be visited by the sampler given  a ``warm'' initialization, mixing time may provide an overly pessimistic estimate for the convergence rate.  
One possible remedy was described in the recent work of \citet{atchade2021approximate}, who proposed to study initial-state-dependent mixing times using restricted spectral gap, a notion that generalizes the spectral gap of a transition matrix, and derived a rapid mixing result for Bayesian variable selection. 
We extend this technique to our setting and show that it can be integrated with the multicommodity flow method to produce sharp mixing time bounds for both random walk and informed MH schemes.

A summary of the mixing time bounds obtained in this work is  presented in Table~\ref{table:overview}. It is important to note that these bounds do not account for the per-iteration computational cost. In practice, an informed MH iteration is typically much more time-consuming; see Remark~\ref{rmk:wall time}.

\begin{table}[t]
\centering
\caption{ Mixing time bounds obtained in this paper. 
Here, $\pi$ is the target posterior distribution,  $\pi_{\min} = \min_x \pi(x)$, $M$ is the maximum neighborhood size, and $L$ is a parameter of the informed proposal scheme. 
Mixing times $\tau$ and $\tau_x$ are defined in~\eqref{eq:def-mix} and~\eqref{eq:def-mix-x} respectively, where  $\epsilon$ is treated as fixed. 
See theorem statements for the required assumptions.  Note that the computational cost per iteration is not considered. 
} 
\resizebox{\columnwidth}{!}{
\begin{tabular}{ccccc}
\hline
\multirow{2}{*}{Setting} & \multirow{2}{*}{Proof techniques}  & \multirow{2}{*}{Random walk MH}     & \multirow{2}{*}{Informed MH} & Require   \\  
& & & & warm start \\ \hline 
 \multirow{3}{*}{Unimodal}   & Path method &   Thm~\ref{th:rwmh-1}: \, $\tau = O\left(\N \log \frac{1}{ \pi_{\min } } \right)$    &   Thm~\ref{th:mix-imh}: \,  $\tau = O\left(\log  \frac{1}{ \pi_{\min } } \right)$  & No \\ 
\cmidrule{2-5}
 & \multirow{2}{*}{Drift condition}  &  \multirow{2}{*}{Not considered (Remark~\ref{rmk:rwmh_drift})
 }     &    \multirow{2}{*}{Thm~\ref{th:mix-imh-drift}: \, $\tau = O\left( \frac{\log (1/ \pi_{\min }) }{\log( L / M)}  \right)$} &  \multirow{2}{*}{No} \\ 
 & & &  &  \\ 
 \hline 
Beyond        & Path method  +       & \multirow{2}{*}{Thm~\ref{th:mix-rwmh-approx}: \, $\tau_x = O\left(\N \log \frac{1}{\pi(x)} \right)$} &  \multirow{2}{*}{Thm~\ref{th:mix-imh-approx}: \, $\tau_x  = O\left(\log \frac{1}{\pi(x)} \right)$ } & \multirow{2}{*}{Yes} \\
unimodal  & restricted spectral gap      &     &         \\ \hline
\end{tabular}
}

\label{table:overview}
\end{table}

\subsection{Organization of the paper}\label{sec:intro-structure}
In Section~\ref{sec:bvs}, we review some recent results about the Bayesian variable selection problem, which will be used as an illustrative example throughout this paper; this section can be skipped for knowledgeable readers.  
Section~\ref{sec:prelim} presents the setup for our theoretical analysis and reviews the mixing time bounds for random walk MH algorithms. 
In Section~\ref{sec:imh}, we consider unimodal target distributions and prove   mixing time bounds for informed MH algorithms via both the path method and drift condition analysis. 
Section~\ref{sec:beyond} generalizes our results to a potentially multimodal setting. 
Simulation studies are conducted in Section~\ref{sec:simul}, and Section~\ref{sec:discussion} concludes the paper with some further discussion. 
Proofs for all results presented in the main text and additional technical details are deferred to the appendix.

\section{Working example: variable selection}\label{sec:bvs}
We review in this section some recent advancements in understanding the complexity of MCMC sampling for high-dimensional spike-and-slab variable selection, one of the most representative examples for discrete-space models in Bayesian statistics~\citep{tadesse2021handbook}.  

\subsection{Target posterior distribution}\label{sec:bvs-background}
We consider the standard linear regression model where a design matrix $\bX \in \bbR^{n \times p}$ and a response vector $\by \in \bbR^{n}$ are assumed to satisfy 
\begin{align}\label{eq:ln.model}
    \by = \bX \bbeta^* + \bz, \quad  \bz \sim \mvn(0, \sigma^2 \bI_n),
\end{align}
where $\bbeta^* \in \bbR^p$ is the unknown vector of regression coefficients and $\mvn$ denotes the multivariate normal distribution.  
Introduce the indicator vector $\delta  = (\delta_1, \dots, \delta_p) \in \{0, 1\}^p$ such that $\delta_j = 1$ indicates that the $j$-th variable has a non-zero effect on the response. 
Variable selection is the task of identifying the true value of $\delta$, that is, finding the subset of  variables with  nonzero regression coefficients.   
Given $\delta$, we can write $\by = \bX_\delta \bbeta_{\delta} + \bz$ with $\bz \sim \mvn(0, \sigma^2 \bI_n)$, where $\bX_\delta$ and $\bbeta_\delta$, respectively, denote the submatrix and subvector corresponding to those variables selected in $\delta$. 
We will also call $\delta$ a model, and $\delta = (0, 0, \dots, 0)$ will be referred to as the empty model. 
Consider the following two choices for the space of allowed models: 
\begin{equation}\label{eq:def-vs}
        \cV = \{0, 1\}^p,  \text{ or }  \cV_s = \{ \delta \in \cV \colon ||\delta||_1 \leq s \}, 
\end{equation} 
where $|| \cdot ||_1$ denotes the $L^1$-norm and   $s$ is a positive integer.  
The set $\cV$ is the unrestricted space, while $\cV_s$ represents a restricted space with some sparsity constraint. 
Let $|\cdot|$ denote the cardinality of a set. 
It is typically assumed in the high-dimensional literature that the sparsity parameter $s$ increases to infinity with $p$, in which case both $|\cV| = 2^p$ and $|\cV_s| = O(p^s)$ grow super-polynomially with $p$. 

Spike-and-slab variable selection is a Bayesian procedure for constructing a posterior distribution over $\cV$ or $\cV_s$, which we denote by $\pi(\delta)$. 
In Section~\ref{subsec:bvs}, we recall one standard approach to specifying the prior distribution for $(\delta, \bbeta, \sigma^2)$, which leads to the following  closed-form expression for $\pi$ up to a normalizing constant: 
\begin{equation}\label{eq:post}
     \quad \pi(\delta) \propto \frac{1}{p^{\kappa ||\delta||_1}} \cdot \frac{(1+g)^{-(||\delta||_1 / 2)}}{\left\{1+g\left( 1 - \mathsf{r}^2(\delta) \right)\right\}^{n / 2}}.  
\end{equation}
In~\eqref{eq:post}, $\kappa, g$ are prior hyperparameters and $\mathsf{r}^2(\delta)= \by^\T (\bX_\delta(\bX_\delta^\T \bX_\delta)^{-1} \bX_\delta^\T) \by / (\by^\T \by)$ is the coefficient of determination.  
Exact calculation of the normalizing constant is usually impossible,  because it would require a summation over the entire parameter space, which involves super-polynomially many evaluations of $\pi$. 

\subsection{MH algorithms for variable selection}\label{sec:bvs-random-walk}
To find posterior probabilities of models of interest or evaluate integrals with respect to $\pi$, the most commonly used method is to use an MH algorithm to generate samples from $\pi$. 
The transition probability from $\delta$ to $\delta'$  in an MH scheme can be expressed by 
\begin{equation}
    \bP(\delta, \delta') = \begin{cases}
    \bK(\delta, \delta') \bA(\delta, \delta'), \quad \quad \quad \text{ if } \delta \neq \delta',\\ 
    1 - \sum\nolimits_{\tilde{\delta}: \tilde{\delta}  \neq \delta} \bP(\delta, \tilde{\delta}), \quad \text{ if } \delta = \delta',
\end{cases} 
\end{equation}
where $\bK(\delta, \delta')$ is the probability of proposing to move from $\delta$ to $\delta'$, and the associated acceptance probability $\bA(\delta, \delta')$ is given by
\begin{equation}\label{eq:accept-prob}
\bA(\delta,\delta') = \min \left\{ 1, \frac{\pi(\delta')\bK(\delta', \delta)}{\pi(\delta)\bK(\delta, \delta')}\right\}, 
\end{equation}
ensuring that $\pi$ is the stationary distribution of $\bP$.

The efficiency of the MH algorithm depends on the choice of $\bK$, and there are many strategies for selecting the proposal neighborhood and assigning the proposal probabilities. Recall that the neighborhood of $\delta$  refers to the support of the  distribution $\bK(\delta, \cdot)$, which we denote by $\cN(\delta)$. Since most algorithms used in practice do not re-propose the current state, we assume $\delta \notin \cN(\delta)$. 
Let us begin by considering random walk proposals such that $\bK(\delta, \cdot)$ is simply the uniform distribution on $\cN(\delta)$.  
To illustrate the importance of selecting a proper proposal neighborhood, we first present two naive choices that are bound to result in slow convergence.

\begin{example}\label{ex:vs1.imh}
Let the state space be $\cV$ and the proposal $\bK(\delta, \cdot)$ be a uniform distribution on $\cN(\delta) = \cV$ for each $\delta$. 
In this case, we get an independent MH algorithm, which is clearly ergodic but usually mixes very slowly. 
For example, suppose the true model is $\delta^* = (1, 0, 0, \dots, 0)$ and the data is extremely informative. 
Even if the chain is initialized at the empty model, it takes on average $2^p$ iterations to propose moving to $\delta^*$.
\end{example} 

\begin{example}\label{ex:vs2.binary} 
Let $T \colon \cV \rightarrow \{0, 1, \dots, 2^{p}-1\}$ be a one-to-one mapping defined by $T(\delta) = \sum_{j = 1}^p \delta_j 2^{j - 1}$. 
Define $\cN(\delta) = T^{-1}(\{ T(\delta) + 1, T(\delta) - 1 \})$. 
In words, we number all elements of $\cV$ from $0$ to $2^p - 1$, and the proposal is a simple random walk on $\{0, 1, \dots, 2^{p}-1\}$. 
The resulting MH algorithm is also ergodic, but again the mixing is slow, since it requires at least $2^p - 1$ steps to move from the empty model to  the full model $(1, 1, \dots, 1)$. 
\end{example}

The neighborhood size is exponential in $p$ in Example~\ref{ex:vs1.imh} and is a fixed constant in Example~\ref{ex:vs2.binary}. Both are undesirable, and it is better to use a neighborhood with size polynomial in $p$. The following choice is common in practice: 
\begin{equation}\label{eq:vs.N1}
    \cN_1(\delta) = \{\delta' \in \cV \colon  ||\delta-\delta'||_1 =1 \}. 
\end{equation} 
The set $\cN_1(\delta)$ contains all the models that can be obtained from $\delta$ by either adding or removing a variable.  
When the design matrix $\bX$ contains highly correlated variables, it is generally considered that $\cN_1$ is too small and introducing ``swap'' moves is beneficial, which means to remove one variable and add another one at the same time. 
The resulting random walk MH algorithm is often known as the add-delete-swap sampler; denote its neighborhood by $\cN_{\rm{ads}}$, which is defined by 
\begin{equation}\label{eq:def-vs-ads}
\begin{aligned}
    \cN_{\rm{ads}}(\delta) = \cN_1(\delta) \cup \cN_{\mathrm{swap}} (\delta), 
    \quad \cN_{\mathrm{swap}} (\delta) = \left\{ \delta' \in \cV \colon  ||\delta' - \delta||_1 = 2,  ||\delta'||_1 = ||\delta||_1 \right\}. 
\end{aligned}
\end{equation} 
We will treat $\cN_{\rm{ads}}$ and $\cN_1$ as neighborhood relations defined on $\cV$. When implementing the add-delete-swap sampler on the restricted space $\cV_s$, one can still propose $\delta'$ from $\cN_{\rm{ads}}(\delta)$ and simply reject the proposal if $\delta' \notin \cV_s$.  
Note that $| \cN_{\rm{ads}}(\delta) \cap \cV_s| = O(ps)$. 

We end this subsection with two remarks. 
First, a popular alternative to MH algorithms is Gibbs sampling. 
Consider a random-scan Gibbs sampler that randomly picks $j \in \{1, 2, \dots, p\}$ and updates $\delta_j$ from its conditional posterior distribution given the other coordinates. 
It is not difficult to see that this updating is equivalent to randomly proposing $\delta' \in \cN_1(\delta)$ and accepting $\delta'$ with probability $\pi(\delta')/ (\pi(\delta) + \pi(\delta'))$.\footnote{ 
In a general sense, this random-scan Gibbs sampler is also an MH algorithm according to the original construction in \citet{hastings1970monte}. 
} 
By Peskun's ordering~\citep{mira2001ordering}, this Gibbs sampler is less efficient than the random walk MH algorithm with proposal neighborhood $\cN_1$~\citep{george1997approaches}.  

Second, for simplicity, we assume in this work that $\bK(\delta, \cdot)$ is a uniform distribution on $\cN(\delta)$ for random walk proposals.  
But in practice, these proposals can be implemented in a more complicated, non-uniform fashion.  
For example, one can first choose randomly whether to add or remove a variable, and then given the type of move, a proposal of that type is generated with uniform probability. 
This distinction has minimal impact on all theoretical results we will develop.

\subsection{Rapid mixing of a random walk MH algorithm}\label{sec:bvs-yang} 
The seminal work of \citet{yang2016computational} considered a high-dimensional setting with $n, p, s \rightarrow \infty$ and $s \log p = o(n)$, and they proved that, under some mild assumptions, the mixing time of the add-delete-swap sampler on the restricted space $\cV_s$ has order $O(p s^2 (n + \kappa s)\log p)$, polynomial in $(n, p, s)$; in other words, the sampler is rapidly mixing. 
To provide intuition about this result and  proof techniques, which will be crucial to the understanding of the theory developed in this work, 
we construct a detailed illustrative example. 

\begin{example}\label{ex:vs3}
Let $p=3$ and $\bX$ be such that $||\bX_j||_2^2 = n$ for $j= 1, 2, 3$, $\bX_1^T \bX_2 = -0.8 n$, $\bX_2^T \bX_3 = -0.6 n$ and $\bX_1^T \bX_3 = 0.9 n$; the three explanatory variables are highly correlated.  
Let $\by$ be generated by 
\begin{align*}
    \by = 1.25 \bX_1 + \bX_2 + \bz,
\end{align*}
where $\bz$ is orthogonal to each explanatory variable, i.e., $\bX_j^\T \bz = 0$ for each $j$, and $||\bz||_2^2 = n$. 
We calculate the un-normalized posterior probabilities by~\eqref{eq:post}  for all models with $n=1,000$, $g=p^3 = 27$ and $\kappa =1$; the values are given in Table~\ref{table:bvs}.

\begin{table}[!b]
\centering
\caption{Log-posterior probabilities of all possible models in Example~\ref{ex:vs3} with $n = 1,000$, $ g = p^3 = 27$ and $\kappa =1 $. 
For the third column, we  set the constant $C = - \log \pi((0, 0, 0))$. For each $\delta$, we indicate if it is a local mode with respect to the given search space and neighborhood relation. }
{
\footnotesize
\begin{tabular}{cccccc}
\hline
\multirow{2}{*}{$\delta$} & \multirow{2}{*}{$1-\mathsf{r}^2(\delta)$}  & \multirow{2}{*}{$C + \log \pi(\delta) $} & Local mode & Local mode  & Local mode \\
 &   &   & w.r.t.  $(\cV, \cN_1)$  &  w.r.t. $(\cV_2, \cN_1)$ &  w.r.t. $(\cV_2, \cN_{\mathrm{ads}})$ \\ \hline \vspace{2mm} 
$(0,0,0)$ & 1 &  0  & No  & No  & No  \\\vspace{2mm}
$(1,0,0)$ & 0.8704  & 63.98  &No  & No  & No   \\\vspace{2mm}
$(0,1,0)$ & 1    & -2.76     &No   & No  & No   \\\vspace{2mm}
$(0,0,1)$  & 0.8236  &   90.46   & No    & No  & No  \\\vspace{2mm}
$(1,1,0)$   & 0.64      &  207.70   &Global mode  & Global mode   & Global mode  \\\vspace{2mm}
$(1,0,1)$   &  0.8219  &  88.69   &No   & No    & No   \\\vspace{2mm}
$(0,1,1)$   &  0.7243    & 148.95   & No  &  \textbf{Yes}   & No  \\\vspace{2mm}
$(1,1,1)$ & 0.64       &  204.90   & No  & --    & --  \\\hline            
\end{tabular}
}
\label{table:bvs}
\end{table}

The true model, $\delta^* = (1, 1, 0)$, is the global mode of $\pi$, which is expected since  $p$ is small but $n$ is large.  
Further, we indicate in Table~\ref{table:bvs} if each model is a local mode with respect to the given search space and neighborhood relation. 
For example, we say $\delta \in \cV_2$ is a local mode with respect to $(\cV_2, \cN_{\rm{ads}})$ if $\pi(\delta) > \pi(\delta')$ for every $\delta' \in \cV_2 \cap \cN_{\rm{ads}}(\delta)$. 
When there is only one local mode (which must be $\delta^*$ in this example), we say $\pi$ is unimodal.    
Table~\ref{table:bvs} shows that $\pi$ is unimodal with respect to $(\cV_2, \cN_{\mathrm{ads}})$, but it is multimodal with respect to $(\cV_2, \cN_{1})$ and the other local mode is $(0, 1, 1)$. 
A graphical illustration is given in  Figure~\ref{fig:nbd}. 
Another interesting observation is that $\pi(\delta)$ does not necessarily increase as $\delta$ gets closer to $\delta^*$. We have  $\pi((0, 0, 1)) > \pi((0, 0, 0)) > \pi((0, 1, 0))$, while the $L^1$ distance from each of the three models to $\delta^*$ is strictly decreasing. This happens due to the correlation structure among the three variables, and in high-dimensional settings, such collinearity is very likely to occur between some variables.  
We will revisit this example in later discussions.  
\end{example}

\begin{figure}[t]
    \centering
    \includegraphics[height=4cm]{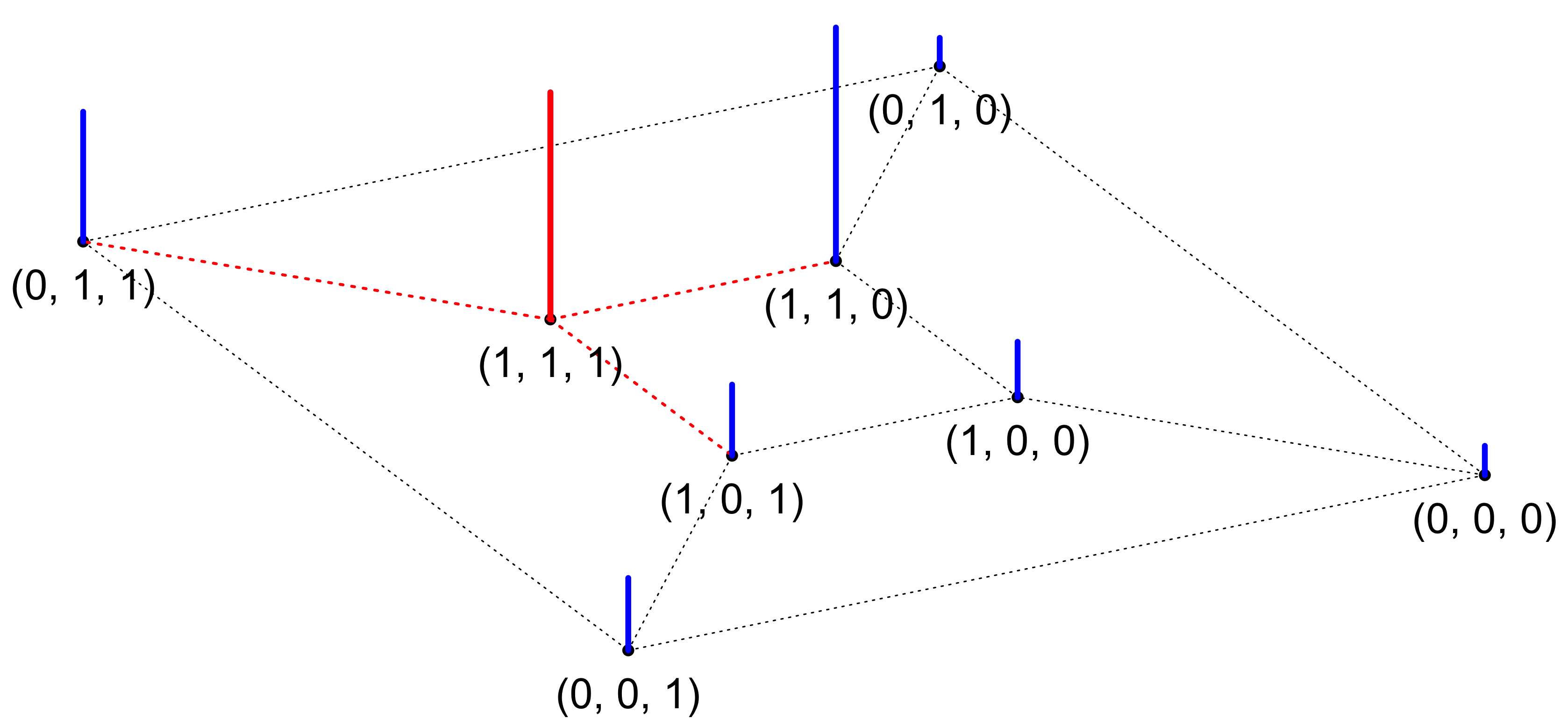} 
    \caption{Visualization of $\pi$ on $(\cV, \cN_1)$ in Example~\ref{ex:vs3}. For every $\delta \in \cV$, the neighboring states $\delta' \in \cN_1(\delta)$ are connected by  dotted lines. 
    The height of each bar is equal to $\sqrt[100]{\pi(\delta)}$.  
    The red bar and dotted lines are removed if the underlying space is  $(\cV_2, \cN_1)$. }
    \label{fig:nbd}
\end{figure}

In Example~\ref{ex:vs3}, $\pi$ is unimodal with respect to $(\cV_2, \cN_{\mathrm{ads}})$. An important intermediate result of~\cite{yang2016computational} was that, with high probability, such a unimodality property still holds in the high-dimensional regime they considered, and the global mode coincides with the true model $\delta^*$. This implies that the add-delete-swap sampler is unlikely to get stuck at any $\delta \neq \delta^*$, since there always exists some neighboring state $\delta'$ such that $\bA(\delta, \delta') = 1$ and thus $\bP(\delta, \delta') = \Omega(p^{-1} s^{-1})$.  
This is the main intuition behind the rapid mixing proof of~\cite{yang2016computational}. 
Interestingly, it was shown in~\cite{yang2016computational} that the unimodality is unlikely to hold on the unrestricted space $\cV$, since local modes can easily occur on the set $\cV \setminus \cV_s$, though  it has negligible posterior mass. 
Since these local modes can easily trap the sampler for a huge number of iterations, they had to consider the restricted space $\cV_s$ in order to obtain a rapid mixing result.  
They also found that without using swap moves, local modes are likely to occur at the boundary of the restricted space. (In Example~\ref{ex:vs3}, $(0, 1, 1)$ is such a local mode.) This is exactly the reason why swaps are required for their rapid mixing proof. 
More generally, for the MCMC convergence analysis of high-dimensional model selection problems with sparsity constraints, examining the local posterior landscape near the boundary seems often a major technical challenge~\citep{zhou2021complexity}. 
Nevertheless, the practical implications of this theoretical difficulty is largely unclear. Even if we choose $\cV$ as the search space, in practice, the sampler is typically initialized within $\cV_s$, and if $s$ is large enough, we will probably never see the chain leave  $\cV_s$, and thus the landscape of $\pi$ on $\cV \setminus \cV_s$ is unimportant.

The above discussion naturally leads to the following question:   can we obtain a   rapid mixing result on the unrestricted space $\cV$ by conditioning on a warm initialization?   
Here, ``warm'' means that the posterior probability of the model is not too small, and we use $\hat{\delta}$ to denote such an estimator, which can often be obtained by a frequentist variable selection algorithm. An affirmative result was given in~\citet{pollard2019rapid}, where a random walk MH algorithm on $\cV$ is constructed such that it proposes models from $\cN_1(\delta)$ and is allowed to immediately jump back to $\hat{\delta}$  whenever $||\delta||_1$ becomes too large.  
They assumed $\hat{\delta}$ was obtained by the thresholded lasso method~\citep{zhou2010thresholded}. 
An even stronger result was obtained recently in~\cite{atchade2021approximate}, who showed that by only using addition and deletion proposals, the random walk MH algorithm on $\cV$ is still rapidly mixing given a warm start. The proof of~\citet{atchade2021approximate} utilizes a novel technique based on restricted spectral gap, which we will discuss in detail later.

\subsection{Rapid mixing of an informed MH algorithm}\label{sec:bvs-informed} 
Now consider informed proposal schemes~\citep{zanella2020informed}, which will be the focus of our theoretical analysis. Unlike random walk proposals which draw a candidate move indiscriminately from the given neighborhood, informed schemes compare the posterior probabilities of all models in the neighborhood and then assign larger proposal probabilities to those with larger posterior probabilities. 
One might conjecture that such an informed MH sampler may behave similarly to a greedy search algorithm and can quickly find the global mode if the target distribution is unimodal. 
But a naive informed proposal scheme may lead to performance even much worse than a random walk MH sampler, as illustrated in the following example.   

\begin{example}\label{ex4:informed_fail}
Consider Example~\ref{ex:vs3} and Table~\ref{table:bvs} again.  Let the informed proposal be given by 
\begin{equation}\label{eq:inform-prop-example-vs}
    \bK(\delta, \delta') =  \pi(\delta')\ind_{\cN_1(\delta)}(\delta') / Z(\delta),
\end{equation}
where $Z(\delta)$ is the normalizing constant; that is, the probability of proposing $\delta' \in \cN_1(\delta)$ is proportional to $\pi(\delta')$.   
Suppose that the chain is initialized at the empty model $\delta_0 = (0,0,0)$. 
A simple calculation yields $\bK(\delta_0, \delta_1) \approx 1 - 3 \times 10^{-12}$ where $\delta_1 = (0,0,1)$. 
But the acceptance probability is 
    \begin{align*} 
    \bA(\delta_0, \delta_1) &= e^{90.46} \times \frac{1 / (1 + e^{148.95} + e^{88.69})}{e^{90.46} /  (e^{90.46} + e^{-2.76} + e^{63.98}))} 
        \approx e^{-58.49} \approx 4 \times 10^{-26}. 
    \end{align*}
Hence, we will probably never see the chain leave $\delta_0$. 
In contrast,  for the random walk MH algorithm with  proposal neighborhood $\mathcal{N}_1$, the probability that the chain stays at $\delta_0$ is less than 1/3.  
\end{example}

In order to avoid scenario similar to Example~\ref{ex4:informed_fail}, 
\citet{zhou2022dimension} proposed to use informed schemes with bounded proposal weights. They constructed an informed add-delete-swap sampler and showed that, under the high-dimensional setting considered by~\cite{yang2016computational}, the mixing time on the restricted space is $O(n)$, thus independent of the problem dimension parameter $p$. 
Since each iteration of informed proposal requires evaluating the posterior probabilities of all models in the current neighborhood, the actual complexity of their algorithm should be $O( p s n )$, comparable to the mixing time of the random walk MH sampler. 
The intuition behind the proof of~\citet{zhou2022dimension}, which was based on a drift condition argument, is similar to that of~\citet{yang2016computational}. 
The unimodality of $\pi$ implies that any $\delta \neq \delta^*$ has one or more neighboring models with larger posterior probabilities, and they showed that, for their sampler, the transition probability to  such a ``good'' neighbor is bounded from below by a universal constant.  

In the remainder of this paper, we develop general theories that extend the results discussed in this section  to arbitrary discrete spaces.

\section{Theoretical setup and preliminary results} \label{sec:prelim}

\subsection{Notation and setup for theoretical analysis}\label{subsec:setup} 
Let $\pi$ denote a probability distribution on a finite space $\cX$. 
Let $\cX$ be endowed with a neighborhood relation $\cN\colon \cX \rightarrow 2^\cX $, and we say $x'$ is a neighbor of $x$ if and only if $x' \in \cN(x)$.
We assume $\cN$ satisfies three conditions: (i) $x \notin \cN(x)$ for each $x$, (ii) $x \in \cN(x')$ whenever $x' \in \cN(x)$, and (iii) $(\cX, \cN)$ is connected.\footnote{``Connected'' means that for any $x, x' \in \cX$, there is a sequence $(x_0 = x, x_1,   \dots, x_k = x')$ such that $x_{i} \in \cN(x_{i-1})$ for  $i = 1, \dots, k$.}  
For our theoretical analysis, we  assume the triple $(\cX, \cN, \pi)$ is given and analyze the convergence of MH algorithms that propose states from $\cN$.  
Define 
\begin{equation}\label{eq:def-neighbor-size}
    \N = \N(\cX, \cN) \coloneqq \max_{x \in \cX} |\cN(x)|. 
\end{equation}
For most high-dimensional statistical problems, $|\cX|$ grows at a super-polynomial rate with respect to some complexity parameter $p$, while $\N$ only has a polynomial growth rate. 
Below are some examples. 
\begin{enumerate}[label=(\roman*)]
    \item \textit{Variable selection.}  As discussed in Section~\ref{sec:bvs}, 
    we can let $\cX = \cV$  and $\cN = \cN_{\mathrm{ads}}$, in which case we have $\N(\cV, \cN_{\rm{ads}}) = O(p^2)$. 
    For the restricted space $\cV_s$, we have $\N(\cV_s, \cN_{\rm{ads}}) = O(ps)$.  
    \item \textit{Structure learning of Bayesian networks.}  
    Given $p$ variables, one can let $\cX$ be the collection of all $p$-node Bayesian networks (i.e.,  labeled directed acyclic graphs). 
    Let $\cN(x)$ be the collection of all Bayesian networks that can be obtained from $x$ by adding, deleting or reversing an edge~\citep{madigan1995bayesian}. 
    We have $\N = O(p^2)$ since the graph has at most $O(p^2)$ edges, while $|\cX|$ is super-exponential in $p$ by Robinson's formula~\citep{robinson1977counting}. 
    \item \textit{Ordering learning of Bayesian networks.}  
    Every Bayesian network is consistent with at least one total ordering of the $p$ nodes such that node $i$ precedes node $j$ whenever there is an edge from $i$ to $j$. 
    To learn the ordering, we can let $\cX$ be all possible orderings of $p$, which is the symmetric group with degree $p$. Clearly, $|\cX| = p!$ is super-exponential in $p$. 
    We may define $\cN(x)$ as the set of all orderings that can be obtained from $x$ by a random transposition, which interchanges any two elements of $x$ while keeping the others unchanged~\citep{chang2022order, diaconis1981generating}. This yields $\N = p(p-1)/2$. 
    \item \textit{Community detection.} Suppose there are $p$ nodes forming $K$ communities. One can let $\cX$ be the collection of all label assignment vectors that specify the community label for each node, and define $\cN(x)$ as the set of all assignments that differ from $x$ by the label of only one node~\citep{zhuo2021mixing}.
    We have $|\cX| = O(K^p)$  and $\N = p(K-1)$. 
    \item \textit{Dyadic CART.}  
    Consider a classification or regression tree problem where the splits are selected from $p$ pre-specified locations. Assume $p = 2^K - 1$ for some integer $K \geq 1$, and consider the dyadic CART algorithm, a special case of CART, where splits always occur at midpoints~\citep{donoho1997cart, castillo2021uncertainty, kim2023mixing}.  
    The search space $\cX$ is the collection of all dyadic trees with depth less than or equal to $K$ (in a dyadic tree, every non-leaf node has 2 child nodes). 
    One can show that $|\cX|$ is exponential in $p$ using the fact that  there are $2^{K-2} = (p+1)/4$ possible splits at depth $K-1$. 
    For $x \in \cX$, we define $\cN(x)$ as the set of dyadic trees obtained by either a ``grow'' or ``prune'' operation; ``grow'' means to add two child nodes to one leaf node, and ``prune'' means to remove two leaf nodes with a common parent node. Then $\N = O(p)$.   
\end{enumerate}

Let $\bP$ be the transition matrix of an irreducible, aperiodic and reversible Markov chain with stationary distribution $\pi$. For all Markov chains we will analyze, $\bP$ moves on the graph $(\cX, \cN)$; that is, $\{x' \colon \bP(x, x') > 0\} = \cN(x) \cup \{x\}$. 
Define the total variation distance between $\pi$ and another distribution $\zeta$ by $||\zeta - \pi||_{\mathrm{TV}} = \sup_{A \subset \cX} | \zeta(A) - \pi(A) |$. 
For $\epsilon \in (0, 1/2)$, let
\begin{equation}\label{eq:def-mix-x}
    \tmix_x (\bP, \epsilon) = \min \{t \in \bbN: ||\bP^t(x, \cdot) - \pi ||_{\mathrm{TV}} \leq \epsilon \}, 
\end{equation} 
which can be seen as the ``conditional'' mixing time of $\bP$ given the initial state $x$. 
Define the mixing time of $\bP$ by 
\begin{equation}\label{eq:def-mix}
     \tmix  (\bP, \epsilon) =  \max_{x \in \cX}   \tmix_x (\bP, \epsilon). 
\end{equation} 
It is often assumed in the literature that $\epsilon = 1/4$, because one can show that $\tmix(\bP, \epsilon) \leq \lceil \log_2 \epsilon^{-1} \rceil \allowbreak \tmix(\bP, 1/4)$ for any $\epsilon \in (0, 1/2)$~\citep[Chap. 4.5]{levin2017markov}. 
However, since such an inequality does not hold for $\tau_x$, we will treat $\epsilon$ as an arbitrary constant in $(0, 1/2)$ in this work. 
The complexity of an MH algorithm can be quantified as the product of its mixing time and the complexity per iteration.

It is well known that mixing time can be bounded by the spectral gap~\citep{sinclair1992improved}. 
Our assumption on $\bP$ implies that it has real eigenvalues $1 = \lambda_1 > \dots \geq  \lambda_{|\cX|} > -1$~\citep[Lemma 12.1]{levin2017markov}. 
The spectral gap is defined as $\mathrm{Gap}(\bP) = 1 - \max \{ \lambda_2, |\lambda_{|\cX|}| \}$  and satisfies the inequality
\begin{equation}\label{eq:ineq-gap-mix}
    \tmix(\bP, \epsilon) \leq \mathrm{Gap}(\bP)^{-1} \log \left\{\frac{1}{\epsilon \, \pi_{\min }}\right\}, 
\end{equation} 
where $\pi_{\min } = \min_{x \in \cX} \pi(x)$. 
The quantity $\mathrm{Gap}(\bP)^{-1}$ is known as the relaxation time. 
In our mixing time bounds, we will always work with the ``lazy'' version of $\bP$,   which is defined by $\bP^{\rm{lazy}} = (\bP + \mathsf{I}) / 2$. 
That is, for any $x \neq x'$, we set $\bP^{\rm{lazy}}(x, x') = \bP(x, x') / 2$. 
Since all eigenvalues of $\bP^{\rm{lazy}}$ are non-negative, $1 - \mathrm{Gap}(\bP^{\rm{lazy}})$ always equals the second largest eigenvalue of $\bP^{\rm{lazy}}$.  
 
\subsection{Unimodal conditions}\label{sec:unimodal}
To characterize  the modality and tail behavior of $\pi$, we introduce another parameter $\R$:  
\begin{align}\label{eq:def-R}
    \R =  \R(\cX, \cN, \pi) \coloneqq  \min_{x \in \cX \setminus \{x^*\}}  \max_{y \in \cN(x)} \frac{\pi(y)}{\pi(x)},  \quad 
    \text{  where }    x^* = \argmax_{x \in \cX} \pi(x).
\end{align} 
If $\argmax$ is not uniquely defined, fix $x^*$ as one of the modes according to any pre-specified rule. 
If $R > 1$, we say $\pi$ is unimodal with respect to $\cN$, since for any $x \neq x^*$, there is some   $x' \in \cN(x)$ such that $\pi(x') > \pi(x)$.  
In our subsequent analysis, we will  consider unimodal targets with $\R > \N$ (for random walk MH) or $\R > \N^2$ (for informed MH).  
These conditions ensure that the tails of $\pi$ decay sufficiently fast.  
To see this, for $k \geq 1$, let $\mathrm{Tail}(k) = \{x \colon \bP^{k - 1}(x, x^*) = 0, \; \bP^{k}(x, x^*) > 0 \}$ denote the set of states that  need $k$ steps to reach $x^*$.  
The unimodal condition implies that for any $x \neq x^*$, there is a path $(x_0 = x, x_1,    \dots,   x^*)$ such that $x_i \in \cN(x_{i-1})$ and $\pi(x_{i-1}) / \pi(x_{i}) \geq R$ for each $i$.  
If $x \in \mathrm{Tail}(k)$, the length of this path is at least $k$, and thus $\pi(x) / \pi(x^*) \leq R^{-k}$. Since the number of states in $\mathrm{Tail}(k)$ is at most $\N^k$, we obtain that   $\pi(\mathrm{Tail}(k)) \leq  \pi(x^*) (\N / \R)^{k}$,  
which decreases to $0$ exponentially fast if $\R > \N$.

For high-dimensional variable selection, the unimodal property with $\R > \N$ has been rigorously established on $\cX = \cV_s$ with $\cN =  \cN_{\mathrm{ads}}$. A careful examination of the proof of~\citet{yang2016computational} reveals that
the proof for  $\R > \N$ and that for $\R > \N^2$ are essentially the same; see the discussion in Section S3 of~\cite{zhou2022dimension}.  
The unimodal condition has also been proved for other high-dimensional statistical problems, including structure learning of Markov equivalence classes~\citep{zhou2021complexity}, community detection~\citep{zhuo2021mixing},  and dyadic CART~\citep{kim2023mixing}.  
It should be noted that how to choose a proper $\cN$ so that $\R > \N$ or $\R > \N^2$ can be a very challenging question (e.g., for structure learning), which we do not elaborate on in this paper. 

In order to achieve rapid mixing,  such unimodal conditions are arguably necessary, since for general multimodal targets, the chain can get trapped at some local mode for an arbitrarily large amount of time.  
We further highlight two reasons why the unimodal analysis is important and not as restrictive as it may seem. 
First, mixing time bounds for unimodal targets are often the building blocks for more general results in multimodal scenarios. One strategy that will be discussed shortly is to use restricted spectral gap.  
Another approach is to apply state decomposition techniques~\citep{madras2002markov, jerrum2004elementary, guan2007small}, which works for  general multimodal targets; see, e.g.,~\citet{zhou2022rapid}.   
Second,  it is helpful to compare our condition $\R > \N$ with log-concavity on continuous spaces. Both conditions imply unimodality and exponentially decaying tails. However, for a log-concave distribution $\pi$ on $\bbR^d$, $\pi$ is increasing along any line segment from an arbitrary point $x$ towards $x^*$. In contrast,  for variable selection, Example~\ref{ex:vs3} shows that $\pi$ can be unimodal but decrease when we move towards the true model $\delta^*$ by flipping a coordinate different from  $\delta^*$. Our condition $\R > \N$ is designed to allow for such unimodal targets.

\begin{remark}\label{rmk:RlessM}
    In general, $R>1$ is not sufficient for ensuring fast convergence of a Markov chain reversible with respect to $\pi$. 
    When $1 < R \leq M$, $\pi$ is still unimodal but with heavier tails, in which case the sampler may need to explore a substantial portion of the whole state space to get sufficiently close to $\pi$ in  TV distance.  To illustrate this, we present a slow mixing example with $M = p + 1$ and $R = p$ in Section~\ref{sec:examples}.   
\end{remark}    
 
For the analysis beyond the unimodal setting, we consider a subset $\cX_0 \subset \cX$ and only impose  unimodality  on $\cX_0$. Let $\cN|_{\cX_0}$ denote $\cN$ restricted to $\cX_0$, that is, 
$\cN|_{\cX_0}(x) = \cN(x) \cap \cX_0$.   
Mimicking the definition of $\R$, we define $\R|_{\cX_0}$ by restricting ourselves to $\cX_0$:  
\begin{align}\label{eq:def-R0} 
    \R|_{\cX_0} =  \R(\cX_0, \cN|_{\cX_0}, \pi) =  \min_{x \in  \cX_0 \setminus \{x^*_0\}}  \max_{x' \in \cN|_{\cX_0}(x) } \frac{\pi(x')}{\pi(x)},  
\end{align}
where $x^*_0 = \argmax_{x \in \cX_0} \pi(x)$.
If $\R|_{\cX_0} > 1$,  we say $\pi$ is unimodal on $\cX_0$ with respect to $\cN$. 
Note that $\R|_{\cX_0} > 1$ also implies that $(\cX_0, \cN)$ is connected. 
In Section~\ref{sec:beyond}, we will study the mixing times of MH algorithms assuming $\R|_{\cX_0} > \N$
 (or $\R|_{\cX_0} > \N^2$) for some $\cX_0$ such that $\pi(\cX_0)$ is sufficiently large.

\subsection{Mixing times of random walk MH algorithms}\label{subsec:rwmh}
We first review the mixing time bound  for random walk MH algorithms obtained in~\citet{zhou2021complexity} under a unimodal condition. 
Recall that we assume the random walk proposal scheme can be expressed as   
$$\bK(x, x') = |\cN(x)|^{-1} \ind_{\cN(x)}(x').$$  
The resulting transition matrix of random walk MH can be written as  
\begin{align}\label{eq:P0}
    \PR (x, x^{\prime} )=\left\{\begin{array}{cc}  
     \min \left\{\frac{1}{|\cN(x)|}, \, \frac{\pi (x^{\prime} )   }{\pi(x)  |\cN(x^{\prime} )| }\right\}, & \text { if } x^{\prime} \in \cN(x), \smallskip \\ 
    1 - \sum_{\tilde{x} \in \cN(x)} \PR(x, \tilde{x}), & \text { if } x^{\prime}=x, \smallskip \\
    0,  & \text{ otherwise. } 
    \end{array}\right. 
\end{align} 

\citet{yang2016computational} used~\eqref{eq:ineq-gap-mix} and the ``canonical path ensemble'' argument to bound the mixing time of the add-delete-swap sampler for high-dimensional  variable selection. 
As observed in~\cite{zhou2022rapid} and~\cite{zhou2021complexity},  this method is applicable to the general setting we consider. 
The only assumption one needs is that the triple $(\cX, \cN, \pi)$ satisfies $\R > \N$,  which ensures that, for any $x \neq x^*$, we can identify a path $(x_0 = x, x_1,  \dots, x_{k-1}, x_k = x^*)$  such that $\pi(x_i) / \pi(x_{i - 1}) > \N$ for each $i$.  
A canonical path ensemble is a collection of such paths, one for each $x \neq x^*$, and then a spectral gap bound can be obtained by identifying  the maximum length of a canonical path and the edge subject to the most congestion. 
It was shown in~\cite{zhou2022rapid} that the bound of~\cite{yang2016computational} can be further improved by measuring the length of each path using a metric depending on $\pi$ (instead of counting the number of edges).  
The following result is a direct consequence of~\eqref{eq:ineq-gap-mix} and Lemma 3 of~\cite{zhou2022rapid}.  

\begin{theorem}\label{th:rwmh-1}
Assume $\rho = \R / \N > 1$. Then, we have $\pi(x^*) \geq 1 - \rho^{-1}$, $\mathrm{Gap}(\bP^{\rm{lazy}}_0)^{-1} \leq c(\rho) \N$ and 
\begin{equation}\label{eq:mix-rwmh-1}
    \tmix(\bP^{\rm{lazy}}_0, \epsilon) \leq    c(\rho) \N \log \left\{\frac{1}{\epsilon \, \pi_{\min }}\right\}, 
\end{equation}
where 
\begin{equation}\label{eq:def-c-rho}
    c(\rho) = \frac{4}{ (1 - \rho^{-1/2})^{3} }. 
\end{equation}
\end{theorem}

\begin{remark}\label{rmk:how-to-prove-rapid-mix}
For  high-dimensional  statistical models, we say that a random walk MH algorithm is rapidly mixing if for fixed $\epsilon \in (0, 1/2)$,  $\tmix(\bP^{\rm{lazy}}_0, \epsilon)$ scales polynomially with some complexity parameter $p$. 
As discussed in Section~\ref{subsec:setup},   $\N$ is typically polynomial in $p$ by construction. 
Hence, to conclude rapid mixing from Theorem~\ref{th:rwmh-1}, it suffices to establish two conditions: 
(i) $\log \pi_{\min }$ is polynomial in $p$, and  
(ii) $\liminf \rho > 1$.  
This line of argument is used in most existing works on the complexity of MCMC algorithms for high-dimensional model selection problems; see~\cite{yang2016computational,zhou2021complexity,zhuo2021mixing} and \cite{kim2023mixing}.  In particular, under common high-dimensional assumptions,  $\rho$ grows to infinity at a rate polynomial in $p$,  which implies $c(\rho) \rightarrow 4$, and by Theorem~\ref{th:rwmh-1} we also have $\pi(x^*) \rightarrow 1$,  which is a consistency property of the underlying statistical model.  
\end{remark} 

\subsection{Application to variable selection and structure learning}\label{sec:vs-sl-one}

To illustrate the application of Theorem~\ref{th:rwmh-1} in high-dimensional Bayesian settings, we derive the mixing time bounds for the variable selection and structure learning problems studied in~\citet{yang2016computational} and~\citet{zhou2021complexity}. 
The posterior distribution for the variable selection problem is given in~\eqref{eq:post}. Detailed descriptions of both problems---including model specifications, neighborhood construction   and high-dimensional assumptions---are provided in Sections~\ref{sec:cor-settings-vs} and~\ref{sec:cor-settings-sl}. 
Here, we only comment on the parameters involved in the mixing time bounds. 
For both problems, $n$ denotes the sample size, $p$ the number of variables, and $\kappa$ the prior parameter controlling model sparsity: 
the prior probability of a model with $m$ variables (in variable selection) or a Markov equivalence class with $m$ edges (in structure learning) is proportional to $p^{-\kappa m}$.   
For variable selection, $s$ denotes the maximum model size on the restricted search space $\cV_s$, as discussed in Section~\ref{sec:bvs}. 
For structure learning, we consider directed acyclic graphs with in-degree bounded by $d_{\rm{in}}$ and out-degree bounded by $d_{\rm{out}}$, where $d_{\rm{in}} + d_{\rm{out}} \leq t_0 \log_2 p$ for some universal constant $t_0 > 0$. 
\begin{corollary}\label{cor:vs_thm1}
     Consider the variable selection problem described in Section~\ref{sec:cor-settings-vs}. 
     Suppose Assumptions (A1) to (A5) hold. For the random walk MH algorithm on   $(\cV_s, \cN_{\rm{ads}})$, 
     \begin{equation*}
        \tau(\mathsf{P}_0^{\text{lazy}}, \epsilon)   = O_p \left(   (  n + \kappa s ) p s  \log p  \right). 
     \end{equation*}
\end{corollary}
\begin{proof}
See Section~\ref{sec:proof-vs}. 
\end{proof} 
\begin{corollary}\label{cor:sl_thm1}
 Consider the structure learning problem described in Section~\ref{sec:cor-settings-sl}.  Suppose Assumptions (B1) to (B4) hold. For the random walk MH algorithm  on $(\mathcal{C}_p(d_{\text{in}}, d_{\text{out}}), \cN_{\rm{ads}})$, 
\begin{equation*}
        \tau(\mathsf{P}_0^{\text{lazy}}, \epsilon)   = O_p \left(  n  p^{t_0 + 3} \kappa  \log p  \right).   
\end{equation*} 
\end{corollary}
\begin{proof}
See Section~\ref{sec:proof-sl}. 
\end{proof} 
Both mixing time bounds are polynomial in $n, p, \kappa$. The bound for variable selection improves that in Theorem 2 of~\citet{yang2016computational} by a factor of $s$, while the bound for structure learning is essentially the same as that in Theorem 5 of~\citet{zhou2021complexity}.

\section{Dimension-free relaxation times of informed MH algorithms}\label{sec:imh} 
\subsection{Informed MH algorithms} \label{sec:informed-intro}
Informed MH algorithms generate proposal moves after evaluating the posterior probabilities of all neighboring states. 
Given a weighting function $h: \bbR^+ \rightarrow \bbR^+$, we define the informed proposal by 
\begin{align}
    \bK_h \left(x, x^{\prime}\right)=\frac{h\left( \frac{\pi\left(x^{\prime}\right) }{ \pi(x)} \right)}{Z_h(x) } \mathbbm{1}_{\mathcal{N}(x)}\left(x^{\prime}\right),  
    \text{ where } Z_h(x)= \sum_{\tilde{x} \in \mathcal{N}(x)} h\left(\frac{ \pi(\tilde{x}) }{ \pi(x)}\right); 
\end{align}
that is,  $\bK_h(x, \cdot)$ draws $x'$ from the set $\cN(x)$ with probability proportional to $h( \pi(x')/\pi(x) )$. 
Intuitively, one wants to let $h$ be non-decreasing so that states with larger posterior probabilities receive larger proposal probabilities. 
Choices such as $h(u) = 1 + u,   \sqrt{u}$ or $1 \wedge u$ were  analyzed in~\cite{zanella2020informed}. 
It was observed in~\cite{zhou2022dimension} and also illustrated by our Example~\ref{ex4:informed_fail}   that for problems such as high-dimensional variable selection, naive choices of $h$ could be problematic and lead to even worse mixing than random walk MH algorithms. 
To gain a deeper insight, we  write down the transition matrix of the induced MH algorithm: 
\begin{equation}\label{eq:informed-P}
    \begin{aligned} 
     \bP_h\left(x, x^{\prime}\right) =\left\{\begin{array}{cc}\bK_h\left(x, x^{\prime}\right) \min \left\{1, \frac{\pi\left(x^{\prime}\right) \bK_h\left(x^{\prime}, x\right)}{\pi(x) \bK_h\left(x, x^{\prime}\right)}\right\}, & \text { if } x^{\prime} \neq x, \\ 1-\sum_{\tilde{x} \neq x} \bP_h(x, \tilde{x}), & \text { if } x^{\prime}=x.\end{array}\right.
\end{aligned}
\end{equation}
We expect that $\bK_h(x, \cdot)$ proposes with high probability some $x'$ such that $\pi(x') \gg \pi(x)$. 
But $\bK_h(x', x)$, which depends on the local landscape of $\pi$ on $\cN(x')$, 
can be arbitrarily small if $h$ is unbounded, causing the acceptance probability of the proposal move from $x$ to $x'$ to be exceedingly small.

The solution proposed in~\cite{zhou2022dimension} was to use some $h$ that is bounded both from above and from below. In this work, we consider the following choice of $h$:
\begin{equation}\label{eq:def-optimal-clip}
    h(u) = \mathrm{clip}(u, \ell, L) \coloneqq \left\{\begin{array}{cc}
      \ell,   &  \text{ if } u < \ell, \\
        u,    & \quad  \text{ if } \ell \leq u \leq L, \\
        L,    &  \text{ if } u > L, 
    \end{array}
    \right.
\end{equation} 
where $\ell < L$ are some constants. Henceforth, whenever we write $\bK_h$ or $\bP_h$, it is understood that $h$ takes the form given in~\eqref{eq:def-optimal-clip}.
We now revisit Example~\ref{ex:vs3} and show that this bounded weighting scheme overcomes the issue of diminishing acceptance probabilities.

\begin{example}\label{ex5:clip}
Consider Example~\ref{ex:vs3}. 
It was shown in Example~\ref{ex4:informed_fail} that if an unbounded informed proposal is used, the algorithm can get stuck at $\delta_0 = (0, 0, 0)$, since it keeps proposing $\delta_1 = (0, 0, 1)$ of which the acceptance probability is almost zero.  
Now consider an informed proposal with $h$ defined in~\eqref{eq:def-optimal-clip} and $\ell = p = 3$, $L = p^2 = 9$. 
By Table~\ref{table:bvs}, this yields proposal probability $\bK_h \left(\delta_0, \delta_1\right) = 3/7$; thus, $\delta_1$ receives larger proposal probability than in the random walk proposal. In contrast to the scenario in Example~\ref{ex4:informed_fail},  the acceptance probability of this move equals $1$, since 
\begin{align*}
\frac{\pi\left(\delta_1\right) \bK_h \left(\delta_1, \delta_0\right)}{\pi(\delta_0) \bK_h \left(\delta_0, \delta_1\right)} 
=\;& e^{90.46} \frac{h(e^{-90.46}) / \{h(e^{-90.46}) + h(e^{58.49}) + h(e^{-1.77})\}}{h(e^{90.46})/\{h(e^{90.46}) + h(e^{-2.76}) + h(e^{63.98})\}  } \\
=\;& e^{90.46} \frac{3/(3+3^2+3) }{3^2/(3^2+3+3^2)} \approx 9 \times 10^{38} > 1. 
\end{align*}
We can also numerically calculate that the spectral gaps of $\bP_0$ (random walk MH) and $\bP_h$ are 0.334 and 0.582, respectively, which shows that the informed proposal accelerates mixing. 
Since $p$ is very small in this example, the advantage of the informed proposal is not significant. 
\end{example}

What we have observed in Example~\ref{ex5:clip} is not a coincidence. The following lemma gives simple conditions under which an informed proposal is guaranteed to have acceptance probability equal to $1$. 

\begin{lemma}\label{lm:accept-prob}
Let  $x' \in \cN(x)$. 
Assume that $Z_h(x) \geq L$ and $\pi(x') / \pi(x) \geq \ell \geq \N$. 
Then,   
\begin{align*}
    \frac{\pi\left(x'\right) \bK_h \left(x', x\right)}{\pi(x) \bK_h \left(x, x'\right)} \geq 1. 
\end{align*}
That is,  an informed proposal to move from $x$ to $x'$ has acceptance probability $1$. 
\end{lemma}
\begin{proof} 
    See Section~\ref{sec:appx-proof}. 
\end{proof}
The assumption $Z_h(x) \geq L$ used in Lemma~\ref{lm:accept-prob} is weak, and it will be satisfied if $x$ has one neighboring state $z$ such that $\pi(z)/\pi(x) \geq L$, which is true if $L \leq \R$. 

\subsection{Dimension-free relaxation time bound via the multicommodity flow method}\label{subsec:imh-path}

Using Lemma~\ref{lm:accept-prob} and Theorem 2 of~\cite{zhou2021complexity}, we can now prove a sharp mixing time bound for informed MH algorithms under the assumption $\R > \N^2$. 

\begin{theorem}\label{th:mix-imh} 
Assume $\R > \N^2$. 
Choose $\ell = \N$ and $\N^2 < L \leq \R$. 
Then,  $\mathrm{Gap}(\bP^{\rm{lazy}}_h)^{-1} \leq 2 c(\tilde{\rho}) $, and 
\begin{equation}\label{eq:mix-imh-1}
    \tmix(\bP^{\rm{lazy}}_h, \epsilon) \leq   2 c(\tilde{\rho})   \log \left\{\frac{1}{\epsilon \, \pi_{\min }}\right\}, 
\end{equation}
where $\tilde{\rho} = L / \N^2$ and $c(\rho)$ is given in~\eqref{eq:def-c-rho}. 
\end{theorem}
\begin{proof}
    See Section~\ref{sec:appx-proof}. 
\end{proof}

\begin{remark}\label{rmk:mix-imh}
  Recall that $\mathrm{Gap}(\bP)^{-1}$ is called the relaxation time, which can be used to derive the mixing time bound by~\eqref{eq:ineq-gap-mix}.     
  If we consider an asymptotic regime where $L, \N^2 \rightarrow \infty$ and $\liminf L / \N^2 > 1$,  by Theorem~\ref{th:mix-imh}, the relaxation time for informed MH algorithms is bounded from below by a universal constant (independent of the problem dimension). 
  In particular, this relaxation time bound is improved by a  factor of $\N$,  compared to that for random-walk MH algorithms in Theorem~\ref{th:rwmh-1}. 
  Since the spectral gap of a transition matrix cannot be greater than 2, the order of the relaxation time bound in Theorem~\ref{th:mix-imh} is optimal, and we say it is ``dimension-free''.  For the dimension-free convergence result of MCMC schemes on continuous state spaces, see, e.g.,~\citet{qin2019convergence, ascolani2024dimension,jin2022dimension, ascolani2024scalability}.
\end{remark}

\begin{remark}\label{rmk:wall time}
The mixing time bound does not take into account the computational cost per iteration. 
When implementing informed MH algorithms in practice, the posterior evaluation of all neighboring states typically incurs a cost proportional to $M$ in each iteration. 
Theorem~\ref{th:mix-imh}   suggests that there exists an informed MH algorithm whose total computational complexity is of the same order as that of random walk MH  algorithms.   
When parallel computing resources are available, one can reduce the cost of informed proposals significantly by parallelizing the posterior evaluation.  
\end{remark}

\begin{remark}\label{rmk:mix-imh-L} 
Lemma~\ref{lm:accept-prob} and Theorem~\ref{th:mix-imh} provide useful guidance on the choice of $\ell, L$ in~\eqref{eq:def-optimal-clip}. 
For most problems, after specifying the proposal neighborhood $\cN$, we can simply set $\ell = \N(\cX, \cN)$, which is typically easy to calculate or bound. 
Regarding the choice of $L$, according to Theorem~\ref{th:mix-imh}, for unimodal targets one may use $L = M^{2 + \xi}$ or $L = (1 + \xi) M^2$ for some small $\xi > 0$; see Corollaries~\ref{cor:vs_thm3} and~\ref{cor:sl_thm3} for examples of usage. 
For multimodal targets, the simulation studies in~\cite{zhou2022dimension} suggest that one may want to choose $\ell, L$ such that the ratio $L/ \ell$ is smaller; in other words, one wants to use a more conservative informed proposal that is not overwhelmingly in favor of the best neighboring state. 
\end{remark}

The proof of Theorem~\ref{th:mix-imh} utilizes the multicommodity flow method~\citep{sinclair1992improved}, which generalizes the canonical path ensemble argument and allows us to select multiple likely transition routes between any two states.  
If one uses the canonical path ensemble, the resulting relaxation time bound for informed MH algorithms will still involve a factor of $\N$ as in Theorem~\ref{th:rwmh-1}.  
A detailed review of the multicommodity flow method is given in Section~\ref{sec:appx-prelim}.

\subsection{Better mixing time bounds via the drift condition}\label{subsec:imh-drift}
The mixing time analysis  we have carried out so far relies on employing path methods to establish bounds on the spectral gap---an approach that is predominant in the literature on finite-state Markov chains.  
One exception was the recent work of~\cite{zhou2022dimension}, who used drift condition to study the mixing time of an informed add-delete-swap sampler for Bayesian variable selection.  
In this section, we show that their method can also be applied to general discrete spaces, 
and it can be used to improve the mixing time bound obtained in Theorem~\ref{th:mix-imh} in an asymptotic setting. 

Recall $x^* = \argmax_{x \in \cX} \pi(x)$. 
The strategy of the proof is to establish a drift condition on $\cX \setminus \{x^*\}$, which means that for some $V \colon \cX \rightarrow [1, \infty)$ and $\lbda \in (0, 1)$, 
\begin{equation}\label{eq:def-drift}
(\bP_h V)(x) \leq \lbda V(x), \quad \text{ for any } x \in \cX \setminus \{x^*\},  
\end{equation}
where $(\bP_h V)(x) = \sum_{x' \in \cX} \bP_h(x, x') V(x').$ 
Then, one can apply Theorem 4.5 of~\cite{jerison2016drift} to obtain a mixing time bound of the order $(1 - \lbda)^{-1} \log (\max_{x \in \cX} V(x))$. 
The main challenge is to find an appropriate $V$ such that $\lbda$ is as small as possible. 
After several trials, we find that the choice, 
\begin{equation}\label{eq:optimal-V}
    V(x) = \pi(x)^{1 /\log  \pimin } = \exp\left( \frac{ \log \pi(x)  }{\log  \pimin} \right), 
\end{equation} 
yields the desired mixing rate given in Theorem~\ref{th:mix-imh-drift} below. 
This drift function is similar to but simpler than the one used in~\cite{zhou2022dimension} for variable selection. 
By definition, $V(x)$ always decreases as $\pi(x)$ increases, since $\pimin < 1$ and $V$ is bounded on $[1, e]$. 
 
\begin{theorem}\label{th:mix-imh-drift} 
Assume $\R > \N^2$.  
Choose $\ell = \N$ and $\N^2 < L \leq \R$. If  
\begin{equation}\label{eq:drift_pimin_condition}  
\log \frac{1}{\pi_\mathrm{min}} \leq 
\frac{L \log (L/M)}{8 (e - 1) M^2  }, 
\end{equation} 
then 
\begin{equation}\label{eq:mix-imh-drift}
     \tmix(\bP^{\rm{lazy}}_h, \epsilon) \leq  \frac{8 \log (2e / \epsilon)} {\log (L / M) }   \log \left( \frac{1 }{\pimin} \right). 
\end{equation}  
\end{theorem}
\begin{proof}
    See Section~\ref{sec:appx-proof}. 
\end{proof}

\begin{remark}\label{rmk:comparison}
The order of the mixing time bound in Theorem~\ref{th:mix-imh-drift}  is typically better than that in Theorem~\ref{th:mix-imh}. 
To see this, let $p$ denote the complexity parameter as introduced in Section~\ref{subsec:setup} and assume $L = p^\omega$ and $M = p^\psi$ for constants $\omega, \psi$ with $\omega > 2 \psi$. 
Then, Theorem~\ref{th:mix-imh-drift} implies $\tmix(\bP^{\rm{lazy}}_h, \epsilon) = O( \log \pimin^{-1} /  \log p  ),$ 
which improves the bound in Theorem~\ref{th:mix-imh} by a factor of $1/\log p$. 
The additional condition given in~\eqref{eq:drift_pimin_condition} is very weak. It holds as long as $\pimin \geq \exp (  - c p^{\omega - 2 \psi} \log p )$ for some sufficiently small constant $c > 0$.  
\end{remark}

\begin{remark}\label{rmk:rwmh_drift}
The drift function given in~\eqref{eq:optimal-V} cannot be used to study the mixing times of random walk MH algorithms under the unimodal setting we consider. 
To see this, suppose $|\cN(x)| = \N$ for all $x \in \cN$ and that there exist $x \in \cX, z \in \cN(x)$ such that (i) 
$\pi(z) / \pi(x) = \R$, and (ii) for any $z' \in \cN(x) \setminus \{z\}$,  $\pi(z') / \pi(x) = 1 / 2$, which implies $\bP_0(x, z') = 1/ (2 \N)$. 
Then,  
\begin{align*}
    (\bP_0 V)(x)      
     & =  \bP_0(x,z)   V(z)    + \sum_{z' \in \cN(x) \setminus \{z\}} \bP_0(x,z')  V(x')     \\
     & =   \frac{1}{\N}V(x)  e^{ \frac{ \log \R }{ \log \pimin} }   + \frac{\N - 1}{  \N } V(x) e^{ \frac{ - \log 2}{ \log \pimin} }. 
\end{align*}
Assuming $\N \rightarrow \infty$ and  $\log \R /  \log \pimin^{-1}  = o(1)$, 
we can rewrite the above equation as 
\begin{align*}
    \frac{ (\bP_0 V)(x) }{V(x)} - 1  
 = & \frac{1}{\N} \left(  e^{ \frac{ \log \R }{ \log \pimin} } - 1 \right) + \frac{\N - 1}{  \N }  \left( e^{ \frac{ - \log 2}{ \log \pimin} } - 1 \right)  \\
 = &  - O\left( \frac{ \log \R }{ \N \log \pimin^{-1}} \right) 
 + O \left(\frac{  \log 2}{ \log \pimin^{-1} }  \right). 
\end{align*}
Hence, as long as $(\log \R) / \N = o(1)$, the right-hand side of the above equation is asymptotically positive. Consequently, $(\bP_0 V)(x)$ is asymptotically greater than $V(x)$, which means that there does not exist $\alpha \in (0, 1)$ satisfying~\eqref{eq:def-drift}.   
\end{remark}

\subsection{Application to variable selection and structure learning}\label{sec:vs-sl-two} 
We now apply Theorems~\ref{th:mix-imh} and~\ref{th:mix-imh-drift} to Bayesian variable selection and structure learning. 
For variable selection, the second mixing time bound in Corollary~\ref{cor:vs_thm3} has the same order as that in Theorem~1 of~\citet{zhou2022dimension}, but the latter was derived under a weaker condition via a more refined drift-condition argument.  
For structure learning, to our knowledge, this is the first result for the complexity of informed MH algorithms. 
\begin{corollary}\label{cor:vs_thm3}
  Consider the variable selection problem described in Section~\ref{sec:cor-settings-vs}.  Suppose Assumptions (A1) to (A5) hold with (A3) replaced by (A3').   For the informed MH algorithm on  $(\cV_s, \cN_{\rm{ads}})$ with parameters $\ell = ps$ and $L = \ell^{2 + \xi}$ for some universal constant $\xi > 0$, 
\begin{equation*}
        \tau(\mathsf{P}_h^{\text{lazy}}, \epsilon)   = O_p \left(( n +   \kappa s ) \log p  \right). 
\end{equation*}  
Further, if $n = o( (ps)^\xi )$, then $ \tau(\mathsf{P}_h^{\text{lazy}}, \epsilon)  = O_p \left( n +   \kappa s  \right). $
\end{corollary}  
\begin{proof}
See Section~\ref{sec:proof-vs}. 
\end{proof} 
\begin{corollary}\label{cor:sl_thm3}
Consider the structure learning problem described in Section~\ref{sec:cor-settings-sl}.  Suppose Assumptions (B1) to (B4) hold with (B3) replaced by (B3'). 
For the informed MH algorithm on  $(\mathcal{C}_p(d_{\text{in}}, d_{\text{out}}),  \cN_{\rm{ads}})$ with parameters $\ell = 3 t_0 p^{t_0 + 2} \log_2 p$ and $L = \ell^{2 + \xi}$ for some universal constant $\xi > 0$, 
\begin{equation*}
        \tau(\mathsf{P}_h^{\text{lazy}}, \epsilon)   = O_p \left(         n p   \kappa \right).   
\end{equation*}  
Further, if $n = o(p^{(t_0 + 2 )\xi - 1})$, 
\begin{equation*}
     \tau(\mathsf{P}_h^{\text{lazy}}, \epsilon)   = O_p \left(   \frac{n p  \kappa }{\log p}   \right).   
\end{equation*}
\end{corollary}
\begin{proof}
See Section~\ref{sec:proof-sl}. 
\end{proof}

\section{Mixing times in a multimodal setting}\label{sec:beyond} 

\subsection{Mixing time bounds via restricted spectral gaps}
As we have seen in Section~\ref{sec:bvs}, when we let $\cX$ be the unrestricted search space for a high-dimensional statistical problem, we usually do not expect that $\pi$ is unimodal on $\cX$.
In particular, for high-dimensional model selection problems, local modes can easily occur at some non-sparse models. 
However, we can often establish the unimodality on some subset $\cX_0 \subset \cX$. Moreover, in reality, we typically initialize the MH algorithm at some state in $\cX_0$ (e.g. a sparse model), and if $\cX_0$ is large enough, we probably will not see the chain leave $\cX_0$ during the entire run. 
This suggests that the behavior of the chain on $\cX \setminus \cX_0$ may not matter if we only care about the mixing of the chain given a good initialization.   
In this section, we measure the convergence of MH algorithms using the initial-state-dependent mixing time $\tmix_x$ defined in~\eqref{eq:def-mix-x} and generalize the previous results to a potentially multimodal setting where $\R|_{\cX_0} > \N$ or $\R|_{\cX_0} > \N^2$. 
Combining the restricted spectral gap argument of~\cite{atchade2021approximate} with the multicommodity flow method, we obtain the following theorem, which shows that the mixing will be fast if the chain has a warm start and $\pi(\cX \setminus \cX_0)$ is sufficiently small. 

\begin{theorem}\label{th:mix-rwmh-approx}
Let $\cX_0 \subset \cX, x_0 \in \cX_0, \eta \in (0, 1)$ be such that  
(i) $\rho = \R|_{\cX_0} / \N > 1$, (ii) $\pi(x_0) \geq \eta$, and (iii) $\pi(\cX_0) \geq 1 - \epsilon^2 \eta^2 / 5$.  
We have 
\begin{equation}\label{eq:mix-rwmh-restricted}
    \tmix_{x_0}(\bP^{\rm{lazy}}_0, \epsilon) \leq    c(\rho) \N \log \left\{\frac{1}{2 \epsilon^2 \eta^2  }\right\}, 
\end{equation}
where $c(\rho)$ is given in~\eqref{eq:def-c-rho}. 
\end{theorem} 
\begin{proof}
    See Section~\ref{sec:appx-proof}. 
\end{proof}

Similarly, we can also use the argument of~\cite{atchade2021approximate} to extend Theorem~\ref{th:mix-imh} to the multimodal setting. 
However, this time we need an additional assumption, given in~\eqref{eq:mix-imh-approx-condition} below, on the behavior of $\pi$ at the boundary of the set $\cX_0$.  

\begin{theorem}\label{th:mix-imh-approx} 
Let $\cX_0 \subset \cX, x_0 \in \cX_0, \eta \in (0, 1)$ be such that $  \R|_{\cX_0} > \N^2$,  $\pi(x_0) \geq \eta$, and  $\pi(\cX_0) \geq 1 - \epsilon^2 \eta^2 / 5$.  
Choose $\ell = \N$ and  $\N^2 < L \leq  \R|_{\cX_0}$. 
If 
\begin{equation}\label{eq:mix-imh-approx-condition}
  \max_{x \in \cX_0} \bK_h(x, \, \cX \setminus \cX_0) \leq \frac{1}{2},   
\end{equation}
we have 
\begin{equation}\label{eq:mix-imh-restricted}
    \tmix_{x_0}(\bP^{\rm{lazy}}_h, \epsilon) \leq    4 c(\tilde{\rho})  \log \left\{\frac{1}{2 \epsilon^2 \eta^2  }\right\}, 
\end{equation}
where $\tilde{\rho} = L / \N^2$ and $c(\rho)$ is given in~\eqref{eq:def-c-rho}. 
\end{theorem}
\begin{proof}
    See Section~\ref{sec:appx-proof}. 
\end{proof}

\subsection{Application to variable selection}
 
We  use variable selection to illustrate how Theorems~\ref{th:mix-rwmh-approx} and~\ref{th:mix-imh-approx} are typically utilized in high-dimensional settings. Set $\cX = \cV$ and $\cX_0 = \cV_s$. 
Since it was already shown in~\cite{yang2016computational} that  $\R|_{\cX_0} > \N $ under certain assumptions, it only remains to verify conditions (ii) and (iii) in Theorem~\ref{th:mix-rwmh-approx}. 
If they can be established for some $x_0 \in \cX_0$ and $\eta > 0$ such that $| \log \eta |$ is polynomial in $p$, then $\tmix_{x_0}(\bP^{\rm{lazy}}_0, \epsilon)$ is also polynomial in $p$.  We prove in the following corollaries that this indeed can be achieved under the setting of~\cite{yang2016computational}.  
\begin{corollary}\label{cor:vs_thm4}
  Consider the variable selection problem described in Section~\ref{sec:cor-settings-vs} and   the random walk MH algorithm on the unrestricted space $(\cV, \cN_{\rm{ads}})$. 
  Suppose Assumptions (A1) to (A5) hold.   If $\delta_0 \in \cV_s$ satisfies  
    $\pi_n(\delta_0) \geq  \sqrt{5} \epsilon^{-1} p^{-\kappa s / 4},$ then 
\begin{equation*}
        \tau_{\delta_0}(\mathsf{P}_0^{\text{lazy}}, \epsilon)   = O_p \left( p \kappa  s^2 \log p  \right). 
\end{equation*}   
\end{corollary}  
\begin{proof}
See Section~\ref{sec:proof-vs}. 
\end{proof} 
\begin{corollary}\label{cor:vs_thm5}
  Consider the variable selection problem described in Section~\ref{sec:cor-settings-vs} and   the informed MH algorithm on the unrestricted space $(\cV, \cN_{\rm{ads}})$ with parameters $\ell = ps$ and $L = \ell^{2 + \xi}$. 
  Suppose Assumptions (A1) to (A5) hold with (A3) replaced by (A3'). If 
\begin{equation} \label{eq:vs-boundary}
\sum_{\delta' \in \cV_s} h\left( \frac{\pi_n(\delta')}{\pi_n(\delta)} \right) \geq \sum_{\delta' \in \cV \setminus \cV_s} h\left( \frac{\pi_n(\delta')}{\pi_n(\delta)} \right)   \; \text{ for every } \delta \in \cV_s \text{ s.t. } \cN_{\rm{ads}}(\delta)   \setminus \cV_s  \neq \emptyset, 
\end{equation}
where $h$ is defined in~\eqref{eq:def-optimal-clip},  
and $\delta_0 \in \cV_s$ satisfies  
    $\pi_n(\delta_0) \geq  \sqrt{5} \epsilon^{-1} p^{-\kappa s / 4},$ then   
\begin{equation*}
        \tau_{\delta_0}(\mathsf{P}_0^{\text{lazy}}, \epsilon)   = O_p \left(      \kappa s \log p  \right). 
\end{equation*}  
\end{corollary}  
\begin{proof}
See Section~\ref{sec:proof-vs}. 
\end{proof}  
The assumption on the initial model $\delta_0$ in both corollaries is relatively weak, since it only requires $\pi_n(\delta_0)$ to be of order $p^{-\kappa s / 4}$, where we recall $s$ is the sparsity parameter tending to infinity ($\kappa$ may also grow to infinity). For comparison, \citet{atchade2021approximate} analyzed a slightly different Bayesian variable selection model and required $\pi_n(\delta_0) \geq  p^{-c}$ for some fixed $c > 0$. The condition stated in~\eqref{eq:vs-boundary} corresponds to the assumption~\eqref{eq:mix-imh-approx-condition} in Theorem~\ref{th:mix-imh-approx}. 
Roughly speaking, \eqref{eq:vs-boundary} is satisfied if, for every $\delta$ on the boundary of $\cV_s$,  the number of desirable moves within $\cV_s$ exceeds those leaving $\cV_s$. Hence, though this condition
seems hard to verify theoretically,  it is expected to hold approximately in practice. \\ 
\indent Establishing analogous results for the structure learning problem appears  significantly more challenging. The restricted search space in~\citet{zhou2021complexity} is constructed by imposing both in-degree and out-degree constraints, making it technically difficult to compare its posterior mass with that of the unrestricted space.

\section{Numerical examples}\label{sec:simul} 

\subsection{Bayesian variable selection}\label{subsec:bvs} 
We consider  the Bayesian variable selection model studied in~\cite{yang2016computational}, for which the marginal posterior distribution $\pi(\delta)$ is given in~\eqref{eq:post}, a function of $\mathsf{r}^2(\delta), ||\delta||_1$ and prior hyperparameters. More details about this model are provided in Section~\ref{sec:cor-settings-vs}.   
We will work with the search space $\cV_n$ defined in~\eqref{eq:def-vs}, which is equivalent to setting $\pi(\delta) = 0$ for any $\delta$ such that $||\delta ||_1 > n$.

Throughout our simulation studies, we set $\kappa = 1$, $g = p^3$, $p = 500$, and sample size $n = 200$.  
We always generate $\by$ by $\by = \bX \bbeta^* + \bz $  with $\bz \sim \mathrm{MVN}(0, \bI_n)$. 
The first 5 elements of $\bbeta^*$ are set to
\begin{align*}
    \bbeta^*_{[5]} = \sqrt{\frac{\log p}{n}} (8,-12,8,8,-12),
\end{align*} 
where we use the notation $[k] = \{1, 2, \dots, k\}$;  all the other elements of $\bbeta$ are set to zero. 
Let $\delta^* $ denote the true model, which satisfies $\delta^*_j = \ind_{ \{ j \in [5] \} }$.   
We let all rows of the design matrix $\bX$ be i.i.d from $\mathrm{MVN}(0,\bm{\Sigma})$, and consider two choices of the covariance matrix $\bm{\Sigma}$: 
\begin{enumerate}[label=(\roman*)]
    \item (moderate correlation) $\bm{\Sigma}_{jk} = e^{-2|j-k|}$ for $j \neq k$; 
   \item (high correlation)  $\bm{\Sigma}_{jk} = e^{-|j-k|/4}$ for $j \neq k$. 
\end{enumerate} 
We set $\bm{\Sigma}_{jj} = 1$ for $j \in [p]$ in both cases.

In all the data sets $(\bX, \by)$ we have simulated, regardless of the simulation setting, we find that $\delta^*$ is always the global mode of $\pi$ with significant posterior mass.      
Hence, we use the number of iterations needed to reach $\delta^*$ as an indicator of the mixing of the chain~\citep{peres2015mixing}. 
We compare three different MCMC methods. 
\begin{enumerate}[label=(\alph*)]
    \item RWMH: the random walk MH algorithm described in Section~\ref{subsec:rwmh}. 
    \item IMH: the informed MH algorithm described in Section~\ref{sec:imh} with $\ell = p$ and $L = p^3$.
    \item IMH-unclipped: the informed MH algorithm described in Section~\ref{sec:imh} with $\ell = 0$ and $L = \infty$. 
\end{enumerate}
     For all three algorithms, we   use $\cN_1$ as the proposal neighborhood; that is, the algorithms only propose the next model by adding or deleting a variable. For each integer $m \in [n] \cup \{0\}$, we use
     \begin{align*}
         \Delta^m = \{\delta \in \cV:  ||\delta||_1 = m\},
     \end{align*}
     to denote the set of all models involving $m$ variables.
     The initial model will be sampled from $\Delta^m$ uniformly for some $m$. 
     We always run RWMH for 10,000 iterations 
     and run each informed MH algorithm for 1,500 iterations. 
We report the following metrics. 
\begin{itemize}  
    \item Success: the number of runs where the algorithm samples the true model. 
    \item $H_{\mathrm{true}}$: the median number of iterations needed to sample the true model for the first time. 
    \item Time: the median wall time measured in seconds.
    \item $T_{\mathrm{true}}$:  the median wall time (in seconds) needed to sample the true model for the first time. 
\end{itemize}   

We first generate one data set $(\bX, \by)$ where the covariance matrix of $\bX$ exhibits moderate correlation. 
We initialize RWMH and IMH at models uniformly sampled from $\Delta^m$ and study the effect of $m$ on the mixing. 
For each choice of $m$, we run the algorithms 100 times.  
Results are presented in Figure~\ref{fig:mixing_bvs} and Table~\ref{table:appendix.simul2}. 
We observe that the mixing of the chain depends on the initialization,  and local modes are likely to occur among non-sparse models. 
As illustrated by Figure~\ref{fig:mixing_bvs}, when the chain is initialized at some $\delta \in \Delta^n = \Delta^{200}$, it always gets stuck and barely moves. 
This is probably because a randomly generated $\delta \in \Delta^n$  is likely to be a local mode on $(\cV_n, \cN_1)$. 
As argued in~\cite{yang2016computational},  $\delta \in \Delta^n$ yields a perfect fit with $\mathsf{r}^2(\delta) = 1$, and for a neighboring model $\delta' \in \cN_1(\delta) \cap \cV_{n}$ to have a larger posterior probability, it must achieve a nearly perfect fit, which is unlikely. 
If we sample the initial model from $\Delta^{m}$ for smaller $m$, the performance of the algorithms gets improved dramatically, as shown in both Figure~\ref{fig:mixing_bvs} and Table~\ref{table:appendix.simul2}. 
This  aligns well with the theory developed in Section~\ref{sec:beyond}: $\pi$ seems to be (at least approximately) unimodal on  $(\cV_s, \cN_1)$ for some $s$ much smaller than $n$, and given a warm start, both RWMH and IMH mix quickly. 
Another observation from  Figure~\ref{fig:mixing_bvs} and Table~\ref{table:appendix.simul2} is that, compared to RWMH, IMH needs a better initialization  to achieve fast mixing. Specifically, for RWMH, we need $m \leq 180$ for most runs to be ``successful'' (i.e., find the true model $\delta^*$), and for IMH, we need $m \leq 170$. 
We believe this is because informed algorithms, which behave similarly to gradient-based  samplers on Euclidean spaces, have a stronger tendency to move to nearest local modes than RWMH.  
Consequently, the performance of IMH is more sensitive to the initialization. 
This observation also partially explains why the  condition~\eqref{eq:mix-imh-approx-condition} is required in Theorem~\ref{th:mix-imh-approx}, which is not needed when we analyze RWMH.

\begin{table}[!b]
\centering
\caption{Results for 100 MCMC runs for one simulated data set with moderately correlated design.  
}
{
\footnotesize

\begin{tabular}{ccccccccc}
\hline
\multicolumn{2}{c}{Initialization}  & $\Delta^{0}$ & $\Delta^{100}$ & $\Delta^{150}$ & $\Delta^{170}$ & $\Delta^{180}$ & $\Delta^{190}$ & $\Delta^{200}$   \\ \hline
\multirow{2}{*}{RWMH}  & Success &  99 & 100 & 100 & 100 & 99 & 36 & 0  \\
& $T_\mathrm{true}$   & 1548 & 2322 & 2584 & 2602 & 2620 &-- & -- \\ \hline
\multirow{2}{*}{IMH} & Success & 100 & 100 & 100 & 97 & 7 & 0 & 0   \\
 & $T_\mathrm{true}$    & 5 &   103 & 155 & 179    & -- &-- & --    \\\hline
\end{tabular}
}
\label{table:appendix.simul2}
\end{table}

\begin{figure}[!t]
    \centering 
\includegraphics[height=4.5cm]{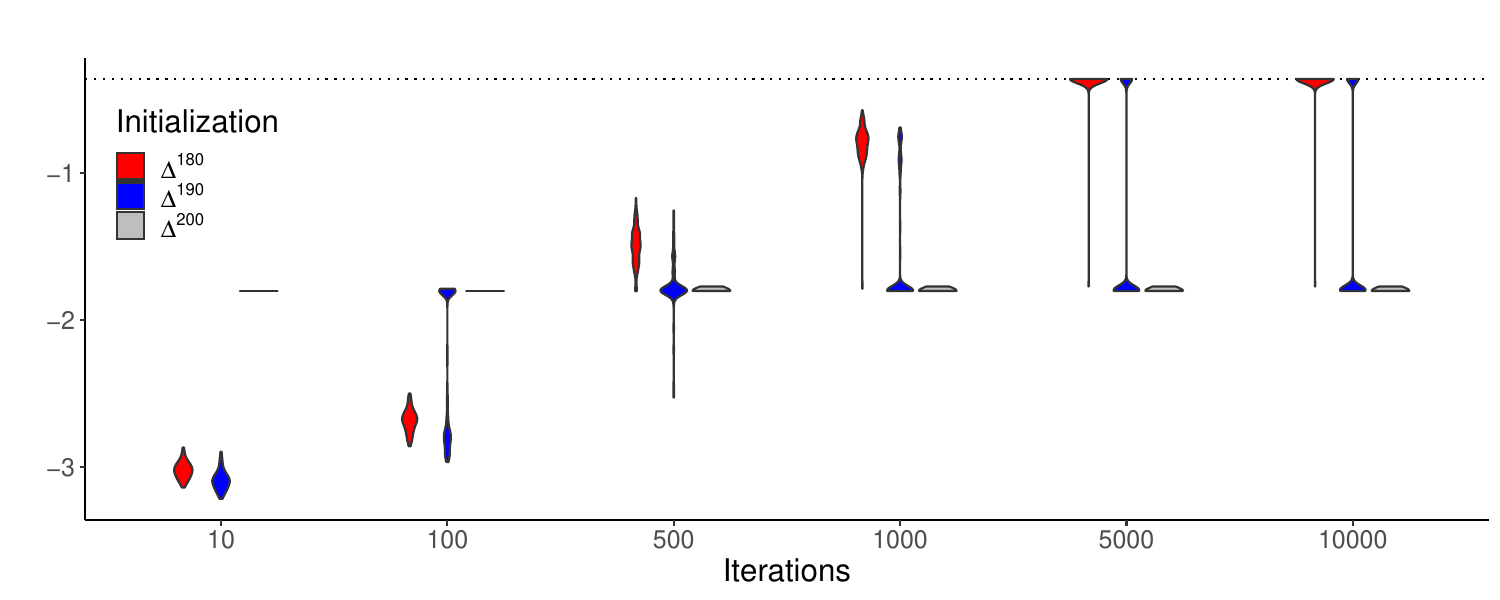}
\includegraphics[height=4.5cm]{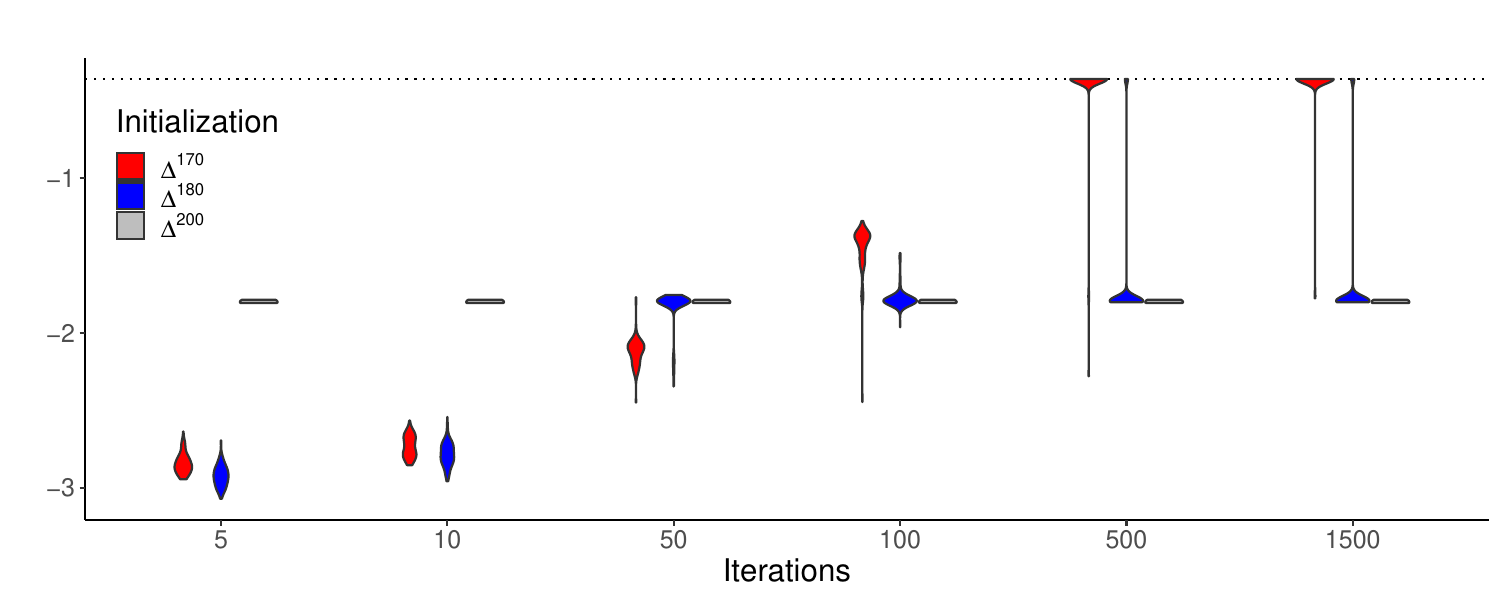}
    \caption{Violin plots visualizing MCMC trajectories for one simulated data set.  
    Each violin gives the distribution of the log-posterior probability (scaled by $10^{-3}$) of the sampled model across 100 runs. 
    The upper panel is for the random walk MH algorithm, and the lower is for the informed MH algorithm. Algorithms are initialized at models uniformly sampled from $\Delta^m$, and each color represents one choice of $m$.
    The black dotted line indicates the scaled log-posterior probability of the true model.
    }
    \label{fig:mixing_bvs}
\end{figure}

\begin{table}[!t]
\centering
\caption{Results for 100 replicates with moderately correlated design
and random initialization from $\Delta^{20}$.} 
{
\footnotesize
\begin{tabular}{cccc}
\hline
 & RWMH & IMH & IMH-unclipped \\
\hline
Success & 100 & 100 & 0 \\
$H_\mathrm{true}$ &1811 &27 & -- \\
Time & 13.6 & 9.3 & 35.3\\
$T_\mathrm{true}$ & 9.1 & 7.5 & -- \\
\hline
\end{tabular}
}
\label{table:rapid}
\end{table}

Next, we consider a more realistic initialization scheme, where all samplers  are started at some $\delta$ uniformly drawn from $\Delta^{20}$. 
We generate 100 replicates of $(\bX, \by)$ under the moderately correlated setting, and compare the performance of RWMH, IMH and IMH-unclipped. 
The result is summarized in Table~\ref{table:rapid}. 
Both RWMH and IMH hit $\delta^*$ within the specified total number of iterations, and IMH achieves this in much fewer iterations. This is expected, since Theorem~\ref{th:rwmh-1} and Theorem~\ref{th:mix-imh-drift} suggest that IMH mixes faster by a factor with the order of $\N$ (for this variable selection problem, $\N = p = 500$).  
We also note that  IMH-unclipped exhibits very poor performance in  Table~\ref{table:rapid}, which is expected according to the reasoning explained in Example~\ref{ex4:informed_fail}.  
It shows that selecting appropriate values for $\ell$ and $L$ in~\eqref{eq:def-optimal-clip} is crucial to the performance of informed MH algorithms in applications like variable selection. 

Simulation results for the highly correlated design are deferred to Section~\ref{sec:supp_bvs}. 
As in the moderately correlated setting, the behavior of the algorithms largely depends on the initialization. However, since $\pi$ tends to be highly multimodal when the design is highly correlated, stronger assumptions on the initialization are needed so that the samplers can quickly find the true model.

\subsection{Bayesian community detection}\label{subsec:sbm} 
Consider the  community detection problem with only two communities, where the goal is to estimate the community assignment vector $z  \in \{1,2\}^{p}$ from an observed undirected graph represented by a symmetric adjacency matrix $A \in \{0,1 \}^{p \times p}$. Here $z_j \in \{1, 2\}$ denotes the community assignment of the $j$-th node.  
We construct a posterior distribution by following the model and prior distribution used in~\cite{zhuo2021mixing}: 
\begin{equation}
    \begin{aligned}
     &  A_{ij} \mid Q, z \stackrel{\text{ind}}{\sim} \mathrm{Bernoulli}(Q_{z_iz_j}),  \; \forall\, i, j \in [p] \text{ and } i < j,  \\
    & Q_{uv} \stackrel{\text{iid}}{\sim} \mathrm{Uniform}(0,1), \quad \forall\, u, v \in \{1,2\} \text{ and } u \leq v, \\
    & \pi_0(z) \propto 1, \quad \forall\, z \in \{1,2\}^{p},
    \end{aligned}
\end{equation}
where $\pi_0(z)$ denotes the prior distribution of $z$ and $Q_{uv}$ is the edge connection probability between one node from community $u$ and another from community $v$. 
By integrating  out $Q$, we obtain a posterior distribution $\pi$ on the finite space $\{1,2\}^{p}$.
Note that due to label switching, we have $\pi(z) = \pi( \tilde{z}  )$, where $\tilde{z}$ is defined by $\tilde{z}_j = 3 - z_j$. 

In the simulation, we set $p = 1,000$ and let the true community assignment vector $z^*$ be given by $z^*_j = 1$ if $1 \leq j \leq p/2$ and $z^*_j = 2$ otherwise. 
We generate the observed graph by sampling each $A_{ij}$ (with $i < j$) independently from $\mathrm{Bernoulli}(p_{\mathrm{within}})$ when $z_i = z_j$ and from $\mathrm{Bernoulli}(p_{\mathrm{between}})$ when $z_i \neq z_j$. 
We use $p_{\mathrm{within}} = 10^{-1}$ and $p_{\mathrm{between}} = 10^{-8}$ so that the two communities are well separated in the observed graph.  
We consider both random walk MH (RWMH) and informed MH (IMH) algorithms, whose proposal neighborhood is defined as $\cN(z) = \{z' : d_{\mathrm{H}}(z, z') = 1\}$, where $d_{\mathrm{H}}$ is the Hamming distance. In words, the proposal only changes the community assignment of one node.  
For IMH, this time we use $\ell = p^{-1}$ and $L = p^3$ (reason will be explained later).  
We simulate 100 data sets, and for each data set we run RWMH for 20,000 iterations and IMH for 2,000 iterations.   
We consider two different initialization schemes, denoted by $z^\mathrm{bad}$ and $z^\mathrm{good}$.
The assignment vector $z^\mathrm{bad}$ is randomly generated such that half of the community assignments are incorrect, while $z^\mathrm{good}$ is generated with only one third of the assignments being incorrect. 
We summarize the results in Table~\ref{table:init_sbm} where the four metrics are defined similarly to those used in Section~\ref{subsec:bvs}. Note that we treat both $z^*$ and $\tilde{z}^*$ as the true model, where $\tilde{z}^*_j = 3 - z^*_j$ for each $j \in [p]$.  

\begin{table}[!h]
\centering
\caption{
Results for 100 replicates of the community detection problem.}  
{
\footnotesize
\begin{tabular}{ccccc}
\hline
&\multicolumn{2}{c}{RWMH} & \multicolumn{2}{c}{IMH} \\ 
Initialization &$z^{\mathrm{bad}}$ & $z^{\mathrm{good}}$ &$z^{\mathrm{bad}}$ & $z^{\mathrm{good}}$ \\
\hline 
Success   &  41 & 100 & 49  & 100\\
$H_\mathrm{true}$ & -- & 10901 & -- & 333 \\
Time  & 8.4 & 6.8 & 76.1 & 75.5\\
$T_\mathrm{true}$ & -- & 5.6  & -- & 13.4\\
\hline
\end{tabular}
}
\label{table:init_sbm}
\end{table}

\begin{figure}[!h]
    \centering
    \includegraphics[height=5cm]{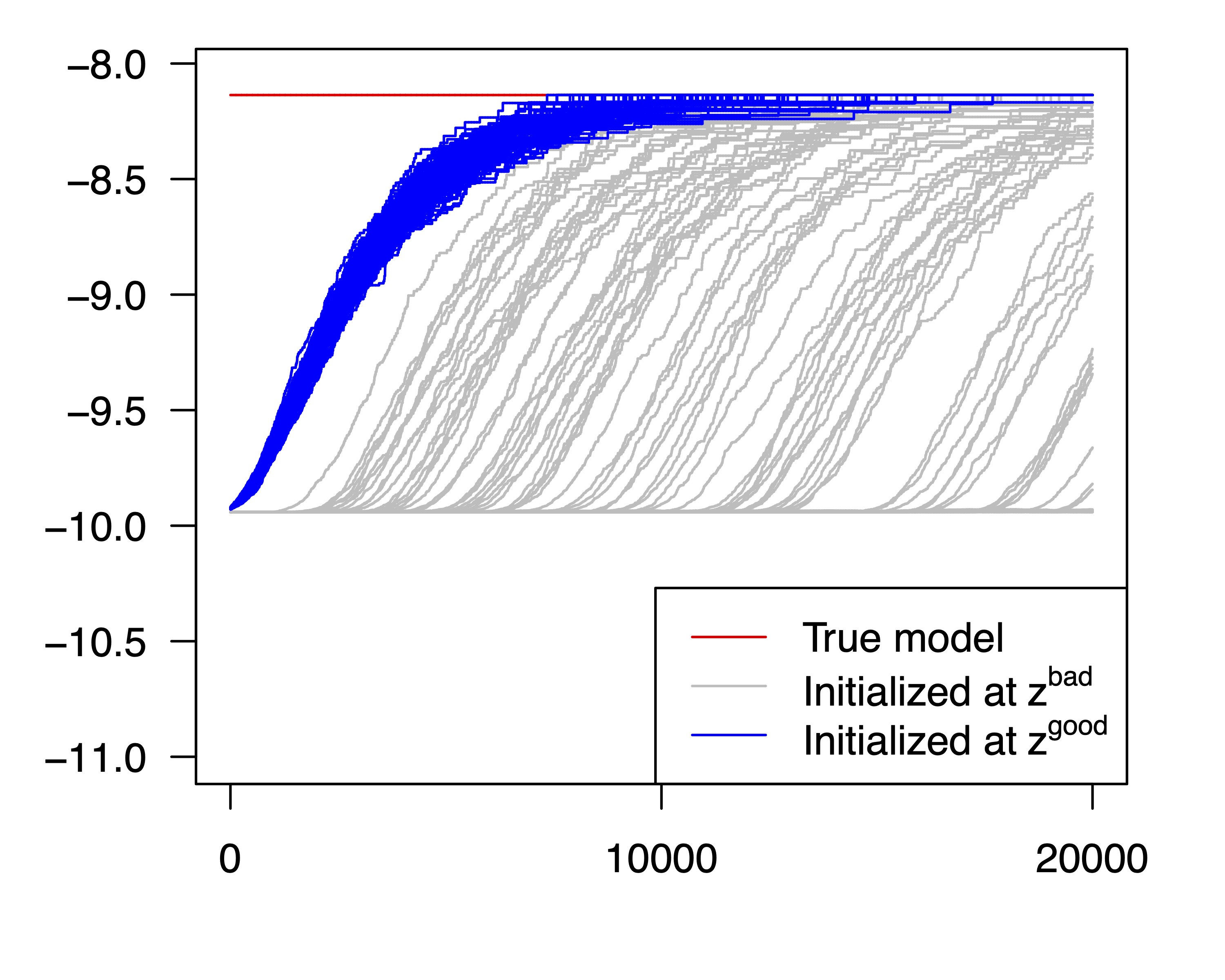}
    \includegraphics[height=5cm]{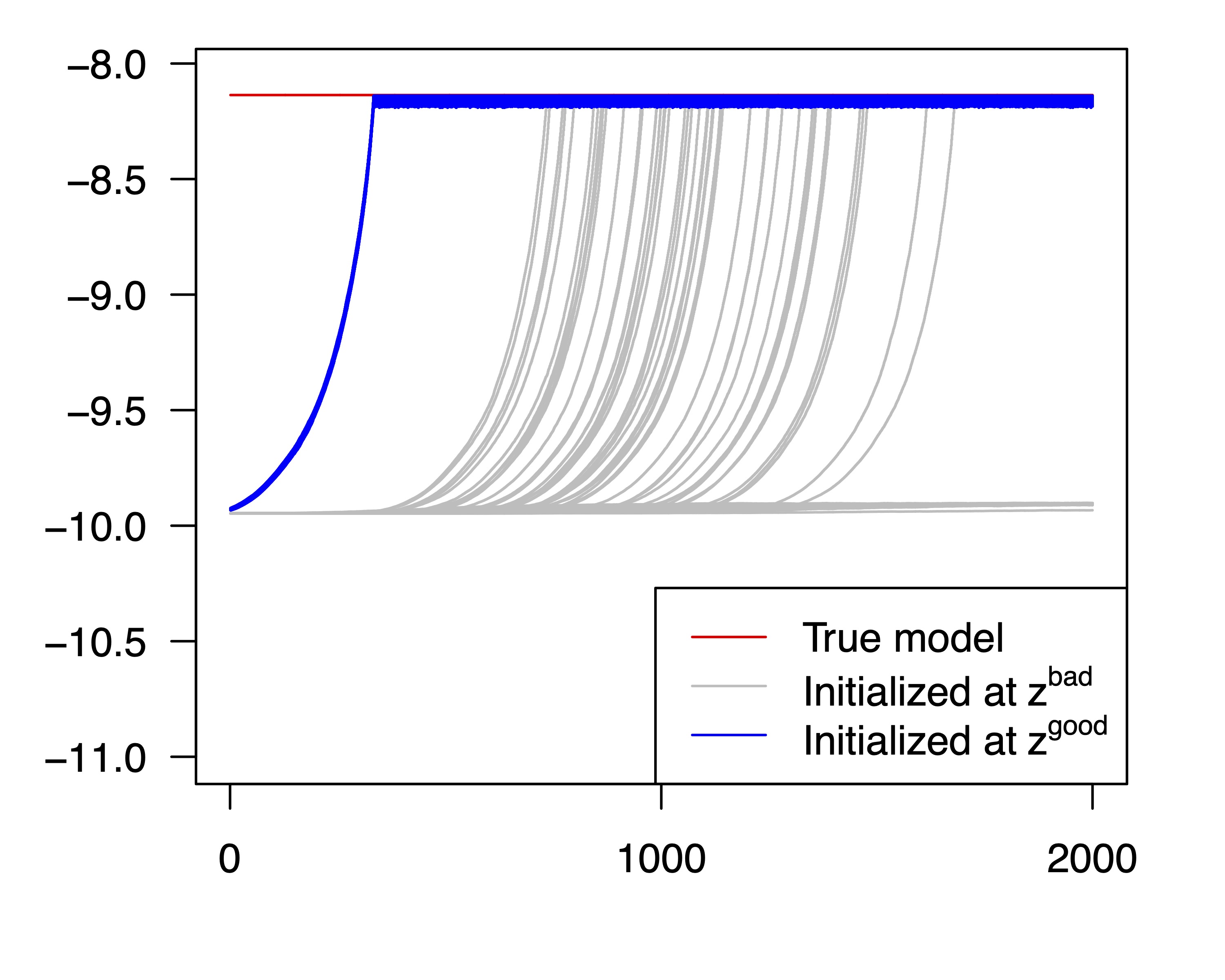} 
    \caption{Log-posterior probability $\times 10^{-4}$ versus the number of iterations in 100 MCMC runs, initialized at $z^{\mathrm{bad}}$  (gray) or $z^{\mathrm{good}}$ (blue). The red line indicates the log posterior probability of the true assignment. 
    The left panel is for RWMH, and the right is for IMH.  
    }
    \label{fig:mixing}
\end{figure} 

The result confirms our theoretical findings  in Section~\ref{sec:beyond} regarding the effect of initialization on the mixing. 
When we initialize the samplers at $z^\mathrm{bad}$, fewer than 50 \% of the runs reach the true assignment vector up to label switching,  while all chains find the true assignment within a reasonable number of iterations when $z^\mathrm{good}$ is used for initialization.  
In Figure~\ref{fig:mixing}, we plot the scaled log-posterior trajectories of 100 MCMC runs for one simulated data set. 
The figure shows that some runs initialized at $z^\mathrm{bad}$ get trapped near it, and when the chain is initialized at $z^\mathrm{good}$, it exhibits much faster mixing.  

We have also tried IMH with $\ell = p$ and $L = p^3$ and observed that in all 100 replicates with  initial state $z^{\mathrm{bad}}$, the sampler fails to find the true assignment vector and gets stuck around $z^{\mathrm{bad}}$. 
We have numerically examined the posterior landscape around $z^{\mathrm{bad}}$ and found that the ratio $\pi(z') / \pi(z^{\mathrm{bad}})$, 
where $z' = \argmax_{z \in \cN(z^{\mathrm{bad}}) } \pi(z)$, is always very small (often around $1$).  Hence, an informed proposal with $\ell = p = 1,000$ only proposes $z'$ with probability $1/p$ (and may get rejected due to small acceptance rates). 
For comparison, using $\ell = p^{-1}$ assigns larger proposal probability to $z'$ and sometimes does help  IMH move away from $z^{\mathrm{bad}}$, as shown in Figure~\ref{fig:mixing}.

\section{Discussion and concluding remarks}\label{sec:discussion} 
Our theory suggests that an MH algorithm converges fast if it is initialized within a set $\cX_0$ such that $\pi(\cX_0)$ is large and $\pi$ is unimodal on $\cX_0$ with respect to the given neighborhood relation $\cN$.   
In practice, it is difficult to determine if $\pi$ is indeed unimodal with respect to $\cN$, and thus a larger proposal neighborhood is sometimes preferred to prevent the chain from getting stuck at local modes. For example, in  variable selection, one can enlarge the proposal neighborhood by allowing adding or removing multiple variables simultaneously, which is a common practice in the literature~\citep{lamnisos2009transdimensional, guan2011bayesian, titsias2017hamming, zhou2019fast, liang2022adaptive, griffin2021search}.  
However, using a too large neighborhood is not desirable either, which is sometimes called the ``Goldilocks principle''~\citep{rosenthal2011optimal}.  
Our theory suggests that we should try to use the smallest neighborhood that  leaves $\pi$ unimodal on $\cX_0$. 
Developing diagnostics related to the unimodal condition could be an interesting direction for future work and provide guidance on choosing the neighborhood size in adaptive methods~\citep{griffin2021search, liang2022adaptive}.

One important technical contribution of this work is that we have demonstrated how to use multicommodity flow method and single-element drift condition to derive mixing time bounds in the context of discrete-space statistical models. 
To our knowledge, both methods have been rarely considered in the literature on high-dimensional Bayesian statistics. 
For random walk MH algorithms, the canonical path ensemble method may already provide a nearly optimal spectral gap bound, but for informed MH algorithms, the use of multicommodity flow is crucial to obtaining the dimension-free estimate of the relaxation time. 
Regarding the drift condition analysis, general theoretical results have been obtained for random walk or gradient-based MH algorithms on Euclidean spaces.  
In particular, the drift function $V(x) = \pi(x)^{-1/2}$ ($\pi$ denotes the Lebesgue density) has been used to prove geometric ergodicity~\citep{roberts1996geometric, jarner2000geometric} and  to derive explicit  convergence rate bounds in the recent work of~\citet{bhattacharya2023explicit}. 
This choice of $V$ is conceptually similar to the one we introduce in~\eqref{eq:optimal-V}, $V(x) = \pi(x)^{1 / \log \pimin}$, where the exponent is chosen deliberately to optimize the mixing time bound. 
It would be interesting to investigate whether one can also improve the mixing time estimates on Euclidean spaces by using a different exponent (though our choice will not be applicable since $\pimin$ becomes zero on unbounded  spaces).  

Both our theory and simulation studies suggest that informed MH algorithms require stronger assumptions on the posterior distribution or initialization scheme to achieve fast mixing. An open question is whether informed samplers still have theoretical guarantees under the regime $\N < \R < \N^2$, which is not covered by our results such as Theorem~\ref{th:mix-imh}. Note that by Theorem~\ref{th:rwmh-1}, random walk MH algorithms mix rapidly in this case.  

Finally, we point out that MH sampling schemes may not be the most efficient way to utilize informed proposals. 
The informed importance tempering algorithm proposed in~\citet{zhou2022rapid} and \citet{li2023importance}, which generalizes the tempered Gibbs sampler of~\cite{zanella2019scalable}, allows one to always accept informed proposals and use importance weighting to correct for the bias. 
It was shown in~\citet{li2023importance} that this technique can also be applied to multiple-try Metropolis algorithms~\citep{liu2000multiple, chang2022rapidly, gagnon2023improving}, making them rejection-free.  
These multiple-try methods enable the calculation of an informed proposal only on a random subset of $\cN(x)$, with the subset's size adjustable according to available computing resources, which makes them highly flexible and easily applicable to general state spaces.

\section*{Acknowledgement} 
The authors were supported by NSF grants DMS-2311307 and DMS-2245591. 
They would like to thank Yves Atchad\'{e} for helpful discussions about restricted spectral gaps, Changwoo J. Lee for comments on the stochastic block model, and the anonymous reviewers for very helpful suggestions. 

\clearpage 
\newpage
\begin{center}
\LARGE{Appendix}
\end{center}

\begin{appendix}

\section{Review of path methods}\label{sec:appx-prelim}

\subsection{Spectral gap and path methods}\label{sec:appx-path}
Let  $\bP \in [0,1]^{|\cX| \times |\cX|}$ denote the transition probability matrix of an irreducible, aperiodic and reversible Markov chain on the finite state space $\cX$. 
Let $\pi$ denote the stationary distribution and define $\pi(f) = \sum_x \pi(x) f(x)$ for any $f \colon \cX \rightarrow \bbR$. 
Recall that we define the \textit{spectral gap} of $\bP$ as 
$$ \mathrm{Gap} (\bP) = 1-\max \{ \lambda_2, |\lambda_{|\cX|} |\},$$ 
and the relation between spectral gap and mixing time is described by the inequality~\eqref{eq:ineq-gap-mix}.   
In words, if $ \mathrm{Gap} (\bP)$ is close to zero, the chain requires a large number of steps to get close to the stationary distribution in total variation distance. 
Path methods are a class of techniques for finding bounds on $\mathrm{Gap} (\bP)$ by examining likely trajectories of the Markov chain. 
They are based on the following variational definition of $\mathrm{Gap}(\bP)$~\citep[Lemma 13.6]{levin2017markov}:  
\begin{equation}\label{eq:def-gap-variational}
    \mathrm{Gap} (\bP) = \min_{  \substack{f \colon \cX \rightarrow \bbR \\  \textit{ s.t. }\mathrm{Var}_{\pi}(f) \neq 0}}  \frac{\cE(f)}{\mathrm{Var}_{\pi}(f)},
\end{equation}
where $\mathrm{Var}_{\pi}(f) = \sum_{x\in \cX} (f(x) - \pi(f) )^2 \pi(x)$, and $\mathcal{E}$ is the \textit{Dirichlet form} associated to the pair $(\bP, \pi)$, defined as
\begin{align}
    \cE(f ) = \langle(\mathsf{I} - \bP)f, f \rangle_{\pi}, 
\end{align}
where the inner product is defined by 
$$\langle f_1, f_2 \rangle_{\pi} = \sum_{x \in \mathcal{X}} f_1(x)f_2(x)\pi(x).$$
We have the following equivalent characterizations of $\cE(f)$ and $\mathrm{Var}_{\pi}(f)$~\citep[Lemma 13.6]{levin2017markov}, 
\begin{align}
    &\cE(f)  = \frac{1}{2} \sum_{x, x' \in \cX} |f(x) - f(x')|^2 \bP(x,x')\pi(x), \label{eq:dirichlet-alt}
    \\ &\mathrm{Var}_\pi(f)  = \frac{1}{2} \sum_{x, x' \in \cX} |f(x) - f(x')|^2 \pi(x)\pi(x'). \label{eq:var-alt}
\end{align}

Given a sequence $\gamma =(x_0 = x, x_1, \dots, x_{k-1}, x_k = x')$, we write $e \in \gamma$ with $e = (x_{j - 1}, x_j)$ for $j = 1, \dots, k$, and we say $e$ is an edge of $\gamma$. 
We say $\gamma$ is a path if $\gamma$ does not contain duplicate edges. 
Let $\barE$ be the edge set for $\bP$, which is defined by  
\begin{equation}\label{eq:def-E}
    \barE =\{  (x, x') \in \cX^2 \colon  x \neq x', \bP(x, x') >0 \}. 
\end{equation} 
Given $E \subset \barE$, let 
\begin{equation}\label{eq:def-E-paths}
    \Gamma_E(x, x') = \{ \text{paths from $x$ to $x'$ with all edges in $E$}\}, 
\end{equation} 
and let $\Gamma_E = \cup_{x \neq x'} \Gamma_E(x, x')$.  
Define $Q \colon \barE \rightarrow (0, \infty)$ by 
\begin{equation}\label{eq:def-Q}
      Q(e) =  \pi(x) \bP(x, x'), \text{ if } e = (x, x'). 
\end{equation} 
We now present Proposition~\ref{prop:gap_bound}, which provides a bound on the spectral gap using paths and generalizes the multicommodity flow method of~\citet{sinclair1992improved}.   

\begin{proposition}[Theorem 3.2.9 in~\cite{gine1996lectures}]\label{prop:gap_bound} 
Let  $E\subset \barE$, $w \colon E \rightarrow (0, \infty)$, and for $\gamma \in \Gamma_E$, define its $w$-length  by 
\begin{equation}\label{eq:def-length}
    |\gamma|_w = \sum_{e \in \gamma} w(e). 
\end{equation} 
Let a function $\phi \colon \Gamma_E \rightarrow [0, \infty)$ satisfy that 
\begin{equation}\label{eq:sum-flow-property}
    \sum_{\gamma \in \Gamma_E(x, x')} \phi(\gamma) = \pi(x) \pi(x'), \quad \text{ for any } x \neq x'. 
\end{equation}
Then $\mathrm{Gap}(\bP) \geq 1/A(E, w, \phi)$, where
\begin{equation}\label{eq:def-congestion-A}
\begin{aligned}
      A(E, w, \phi)= \max_{e\in E} \left\{ \frac{ \sum_{ \gamma \in \Gamma_E\colon e \in \gamma} |\gamma|_w \phi(\gamma)  }{Q(e)w(e)}  
    \right\}. 
\end{aligned}
\end{equation}  
\end{proposition}

In Proposition~\ref{prop:gap_bound},  $w$ is called a weight function, and $\phi$ is called a flow function.   
Proposition~\ref{prop:gap_bound} directly follows from~\eqref{eq:dirichlet-alt}, ~\eqref{eq:var-alt} and the following lemma, which will be used later for proving a similar result for restricted spectral gap in Section~\ref{sec:appx-spec-gap} (see Proposition~\ref{prop:restricted_gap_bound}).   

\begin{lemma}\label{lemma:gap_bound} 
Let  $E\subset \barE$, $w \colon E \rightarrow (0, \infty)$ and $\phi \colon \Gamma_E \rightarrow [0, \infty)$. 
For any $f \colon \cX \rightarrow \bbR$, we have 
\begin{align*}
      \sum_{x,x' \in \cX} |f(x') - f(x)|^2 \Phi(x, x')   
    \leq \;  A(E, w, \phi) \sum_{(x, x') \in E}  Q( (x, x') )   |f(x') - f(x)|^2, 
\end{align*}  
where $A(E, w, \phi)$ is given by~\eqref{eq:def-congestion-A} and 
\begin{align*}
    \Phi(x, x') = \begin{cases}
         \sum_{\gamma \in \Gamma_E(x, x')} \phi(\gamma), & \text{ if } \,  \Gamma_E(x, x') \neq \emptyset,  \\
        0, & \text{ otherwise}.
    \end{cases} 
\end{align*} 
\end{lemma}

\begin{proof}
Given $e = (z, z')$, let $\d f (e) = f(z') - f(z)$ denote the increment of $f$ along $e$. 
For any $\gamma \in \Gamma_E(x, x')$,   we can write
\begin{align*}
    f(x') -f(x)  = \sum_{e   \in \gamma} \d f(e) 
       =  \sum_{e   \in \gamma} \sqrt{w(e)} \frac{ \d f(e)}{ \sqrt{w(e)} }. 
\end{align*} 
By the Cauchy-Schwarz inequality,
\begin{align*}
       |f(x') -f(x)|^2   
    \leq |\gamma|_w \sum_{e  \in \gamma} \frac{ (\d f(e) )^2}{w(e)}. 
\end{align*}
Hence, 
\begin{align*}
    \sum_{\gamma \in \Gamma_E(x,x')} \phi(\gamma) |\gamma|_w \sum_{e \in \gamma} \frac{(\d f(e) )^2 }{w(e)}  \geq |f(x') -f(x)|^2    \Phi(x, x'). 
\end{align*}
We sum over $x,x' \in \cX$, which yields
\begin{align*} 
          \sum_{x,x' \in \cX} |f(x') - f(x)|^2 \Phi(x, x') 
        \leq \;&   \sum_{\gamma \in \Gamma_E} \phi(\gamma) |\gamma|_w \sum_{e \in \gamma} \frac{(\d f(e) )^2 }{w(e)} \\ 
            = \;& \sum_{e   \in E}  \frac{(\d f(e) )^2}{ w(e)} \sum_{\substack{\gamma\,:\,\gamma \in \Gamma_E \\ \textit{ s.t. } e \in \gamma}}  
       |\gamma|_w \phi(\gamma)       \\
        = \;& \sum_{e   \in E}\left\{ \frac{1}{Q(e)w(e)} \sum_{\substack{          \gamma\,:\,\gamma \in \Gamma_E \\ \textit{ s.t. } e \in \gamma}} 
       |\gamma|_w \phi(\gamma) \right\}     Q(e)(\d f(e) )^2 \\
  \leq\;& A(E, w, \phi)  \sum_{e  \in E}  Q(e)(\d f(e) )^2, 
    \end{align*} 
which completes the proof. 
\end{proof}

\begin{remark} \label{rmk:general-Q-flow-bound}
Observe that in Lemma~\ref{lemma:gap_bound}, one can  replace $Q$ with any function that maps from $E$ to $(0, \infty)$.  
The conclusion still holds and the proof is identical.  
\end{remark}

\begin{proof}[Proof of Proposition~\ref{prop:gap_bound}]
By~\eqref{eq:var-alt}, 
\begin{align*}
    \sum_{x,x' \in \cX} |f(x') - f(x)|^2 \Phi(x, x')   = 2 \mathrm{Var}_\pi(f), 
\end{align*}
and by~\eqref{eq:dirichlet-alt}, 
\begin{align*}
    \sum_{(x, x') \in E}  Q( (x, x') )   |f(x') - f(x)|^2  \leq 2 \cE(f). 
\end{align*}
The result then follows from~\eqref{eq:def-gap-variational}. 
\end{proof}

\subsection{Review of the mixing time bound in~\texorpdfstring{\cite{zhou2021complexity}}{Zhou et al. 2021}}\label{subsec:review}

In this section, we review the proof techniques used in~\cite{zhou2021complexity} to derive the mixing time bound under the unimodal condition. Similar arguments will be used later in Section~\ref{sec:appx-flow-restricted-gap-unimodal} for proving some results of this work.   
Let the triple $(\cX, \cN, \pi)$ satisfy $\R (\cX, \cN, \pi) > \N(\cX, \cN),$
where $\R$ and $\N$ are defined in the main text, 
and recall that we always assume 
\begin{align*}
   \cN(x)  = \{ x' \in \cX \colon \bP(x, x') > 0\}. 
\end{align*}
We begin by constructing the functions $\phi$ and $w$ for utilizing Proposition~\ref{prop:gap_bound}.   

\subsubsection{Constructing a flow \texorpdfstring{$\phi$}{}  }\label{subsubsec:flow} 
Assume that $\pi$ is unimodal on $(\cX, \cN)$.  
We first present a general method for constructing flows by identifying likely paths leading to $x^* = \arg \max_{x \in \cX} \pi(x)$, and derive a simple upper bound  on the maximum load of any edge. 
For this result, we only need to require $\R (\cX, \cN, \pi) > 1$. 

\begin{lemma}[Lemma~B2 of~\cite{zhou2021complexity}] \label{lm:def_flow}    
Suppose   $(\cX, \cN, \pi)$ satisfies  $\R = \R (\cX, \cN, \pi) > 1$. 
Let $S \in (1, \R]$.   
Define 
\begin{equation}\label{eq:condition1}
      \cN_S(x) = \left\{ x' \in \cN(x): \pi(x') / \pi(x) \geq S \right\}, 
\end{equation}
and let 
\begin{equation}\label{eq:def-ES}
    E_S = \{ (x, x') \colon x \in \cN_S(x') \text{ or } x' \in \cN_S(x) \}
\end{equation}
be the resulting edge set. 
Write $\Gamma_S = \Gamma_{E_S}$. 
There exists a flow $\phi\colon  \Gamma_S \rightarrow [0, \infty)$ such that for any $x \neq x'$, 
\begin{equation}\label{eq:flow-sum}
     \sum_{\gamma \in \Gamma_S(x, x')} \phi(\gamma) = \pi(x) \pi(x'),  
\end{equation}
and for any $x \neq x'$ and $(z,z') \in E_S$ with $z' \in \cN_S(z)$,  
    \begin{align}\label{eq:prop_flow}
    \sum_{\substack{\gamma\colon \gamma  \in \Gamma_S(x, x')  \\ \text{ s.t. }  (z,z') \in \gamma} } \phi(\gamma)  \leq \frac{\bP(z, \,z')}{\bP(z, \, \cN_S(z))} \pi(x) \pi(x'). 
    \end{align}   
\end{lemma}

\begin{proof} 
Define an auxiliary transition matrix $\bP_S$ on $\cX \times \cX$ by
\begin{equation}\label{eq:aux}
\begin{aligned}
\bP_S(x,x') 
    =\;& 
    \begin{cases}
        \frac{\bP(x, \,x')}{\bP(x, \, \cN_S(x))} \ind_{\cN_S(x)}(x'), &  \text{if  } \cN_S(x) \neq \emptyset,  \vspace{0.2cm} \\
        0, & \text{if  }\cN_S(x) = \emptyset, \vspace{0.2cm}  \\
        1 - \sum_{\tilde{x} \neq x} \bP_S(x,\tilde{x}), & \text{if  } x = x'.
    \end{cases}
\end{aligned}
\end{equation} 
Given $\gamma=\left(x_0, x_1, \ldots, x_k\right)$, let $\stackrel{\leftarrow}{\gamma} =(x_k, x_{k-1}, \dots, x_0)$ be the reversed path of $\gamma$.
We first construct a normalized flow function, denoted by $f^S: \Gamma_S \rightarrow [0,1]$, as follows.   
\begin{enumerate}[label=(\arabic*), ref=(\arabic*)]
 \item If $x^*$ does not occur in $\gamma$ or $x^*$ occurs at least twice in $\gamma$, let $f^S(\gamma)=0$. 
 \item If $x_k=x^*$, let $f^S(\gamma)=\prod_{i=1}^k \bP_S\left(x_{i-1}, x_i\right)$.
 \item If $x_0=x^*$, let $f^S(\gamma)=f^S(\stackrel{\leftarrow}{\gamma})$.
 \item Define $f^S(\gamma)=f^S\left(\gamma_1\right) f^S\left(\gamma_2\right)$ if $x_j=x^*$ for some $1 \leq j \leq k-1$, where $\gamma_1=\left(x_0, \ldots, x_{j-1}, x^*\right)$ and $\gamma_2=\left(x^*, x_{j+1}, \ldots, x_k\right)$.
\end{enumerate}

Then, the following properties hold for $f^S$ by Lemma~B2 in~\cite{zhou2021complexity}.
For any $x,x' \in \cX$ with $x \neq x'$, we have
\begin{align}\label{eq:def_flow}
    \sum_{\gamma \in \Gamma_S(x, x')} f^S(\gamma) = 1,
\end{align}
and for $e = (z,z')$ with $z' \in \cN_S(z)$, 
\begin{align}
\sum_{\substack{\gamma\colon \gamma \in \Gamma_S(x, x')    \\ \textit{ s.t. } e = (z,z') \in \gamma} } f^S(\gamma)  \leq \bP_S(z, z'). 
\end{align} 
The flow $\phi$ is then obtained by 
\begin{equation}\label{eq:normalize-flow}
    \phi(\gamma) = f^S(\gamma)\pi(x)\pi(x'),  \text{ if } \gamma \in \Gamma_S(x,x'), 
\end{equation} 
which satisfies~\eqref{eq:sum-flow-property} and~\eqref{eq:prop_flow}.
\end{proof}

\begin{remark}\label{rmk:flow.generalize}
This result can be generalized by replacing $\cN_S$ with any $\tilde{\cN}$ such that 
\begin{enumerate}[label=(\roman*)]
    \item $\tilde{\cN}(x) \neq \emptyset$ for any $x \neq x^*$, 
    \item $\pi (y) > \pi(x)$ whenever $y \in \tilde{\cN}(x)$. 
\end{enumerate}
One can then define a transition matrix $\tilde{\bP}$ analogously to $\bP_S$, which is a Markov chain with absorbing state $x^*$. 
The same argument then completes the proof.  
\end{remark}
 
\subsubsection{Defining a weight function 
\texorpdfstring{$w$}{w} }\label{subsubsec:weight}
Consider $E_S$ defined in Lemma~\ref{lm:def_flow}. 
For $e = (z,z') \in E_S$ such that $z' \in \cN_S(z)$, we define 
\begin{equation}\label{eq:def-weight-q}
    w(e) = w(\stackrel{\leftarrow}{e}) = \pi(z)^{-q},
\end{equation} 
where $q \in (0,1)$ is a constant and $\stackrel{\leftarrow}{e} = (z', z)$.  
Let $\Gamma_S$ be as defined in Lemma~\ref{lm:def_flow} with $S > 1$. 
For any $\gamma \in \Gamma_S(x,x')$ with $\phi(\gamma) > 0$, we have
\begin{align}\label{eq:length}
    |\gamma|_w \leq \frac{\pi(x)^{-q} + \pi(x')^{-q}}{1 - S^{-q}}.
\end{align}
To obtain~\eqref{eq:length}, we first decompose $\gamma$ to two sub-paths $\gamma_1, \gamma_2$ such that $\gamma_1 \in \Gamma(x,x^*)$ and $\gamma_2 \in \Gamma(x^*, x')$, which exists by our construction of $\phi$. 
Letting $\gamma_1 = (x = x_0, x_1, \dots, x_k = x^*)$, we obtain by the geometric series calculation that 
\begin{align*}
    |\gamma_1|_w = \sum_{j=0}^{k-1} \pi(x_j)^{-q} \leq  \sum_{j=0}^{k-1} \pi(x)^{-q} S^{-q} \leq \frac{\pi(x)^{-q}}{1 - S^{-q}}.
\end{align*}
Calculating $|\gamma_2|_{w}$ in the same manner and using $|\gamma|_{w} = |\gamma_1|_{w}+|\gamma_2|_{w}$, we get~\eqref{eq:length}. 

\subsubsection{Spectral gap bound }\label{subsub:gap} 
Using $\phi$ and $w$ constructed above, we can now derive a spectral gap bound. 
\begin{lemma}\label{lm:bounding_A}
    Suppose the triple $(\cX, \cN, \pi)$ satisfies  
    $$ \R = \R (\cX, \cN, \pi) > \N(\cX, \cN) = \N , $$ 
    and let $x^* = \argmax_{x \in \cX} \pi(x)$. 
    Let $E_S, \phi$ be as given in Lemma~\ref{lm:def_flow} with $S \in (\N, \R]$ 
    and $w$ be given by~\eqref{eq:def-weight-q}. 
    Then,  
    \begin{equation}\label{eq:final_bound_standard}
         A(E_S, w, \phi) \leq \frac{ c(S/M)}{2 }  \max_{z \in \cX \setminus \{x^*\}} \frac{1}{\bP(z,\cN_S(z))}, 
    \end{equation}
    where $A$ is given by~\eqref{eq:def-congestion-A} and 
    $c(u) = 4(1 - u^{-1/2})^{-3}.$
\end{lemma}
\begin{proof}
We first fix $e = (z,z') \in E_S$ with $z' \in \cN_S(z)$. The case with $z \in \cN_S (z')$ can be analyzed in the same manner by symmetry.  
Observe that 
\begin{align*}
     \frac{1}{Q(e)w(e)} \sum_{\substack{\gamma \colon \gamma \in \Gamma_S \\ \textit{s.t. } e \in \gamma} } \phi(\gamma) |\gamma|_w  
    =\; \frac{1}{Q(e)w(e)} 
    \sum_{x \neq x'} \sum_{\substack{ \gamma \colon \gamma \in \Gamma_S(x, x') \\ \textit{s.t. } e \in \gamma} }\phi(\gamma) |\gamma|_w. 
\end{align*}
For fixed $x, x' \in \cX$ with $x\neq x'$, 
\begin{align*}
     \frac{1}{Q(e)w(e)} \sum_{ \substack{ \gamma \colon \gamma \in \Gamma_S(x, x') \\ \textit{s.t. } e \in \gamma} } \phi(\gamma) |\gamma|_w 
    \leq \; & \frac{\pi(x)^{-q} + \pi(x')^{-q}}{\pi(z)^{1 -q}(1 - S^{-q})\bP(z,z')} \sum_{\substack{\gamma\colon \gamma \in \Gamma_S(x, x')  \\ \textit{ s.t. } e   \in \gamma} }\phi(\gamma) \\
    \leq \; &  \frac{\pi(x)^{1-q}\pi(x') +\pi(x) \pi(x')^{1-q}}{\pi(z)^{1 -q}(1 - S^{-q})}   \frac{1}{\bP(z,\cN_S(z))},
\end{align*}
where the first inequality is from~\eqref{eq:length} and the last inequality is due to~\eqref{eq:prop_flow} with the definition in~\eqref{eq:aux}. 

Let 
$\Lambda(z) = \{x \colon \pi(x) < \pi(z),\; \exists \gamma \in \Gamma_S(z, x) \text{ s.t.  } 
\phi(\gamma) > 0  \} \cup \{z\}$
denote the set of ancestors of $z$ in the graph $(X, E_S)$ (each edge is directed towards the state with larger posterior probability).  
Note that $$\sum_{\gamma \in \Gamma_S(x, x')\colon e \in \gamma } \phi(\gamma) > 0$$ 
only if $x$ or $x'$ belongs to $\Lambda(z)$.  
Hence, summing over distinct $(x, x')$ and using $\pi(x) \leq 1$ for any $x$, we get  
\begin{align*}
      \frac{1}{Q(e)w(e)} \sum_{ \substack{ \gamma \colon \gamma \in \Gamma_S \\ \textit{s.t. } e \in \gamma}} \phi(\gamma) |\gamma|_w  
    \leq\;&  \sum_{x \in \Lambda(z)} \sum_{x' \in \cX}  \frac{ 2\pi(x)^{1-q}\pi(x')^{1-q}  }{\pi(z)^{1 -q}(1 - S^{-q})}   \frac{1}{\bP(z,\cN_S(z))}. 
\end{align*}
Using $S > \N$, a routine geometric series calculation gives 
\begin{align*} 
 \sum_{x \in \Lambda(z)} \sum_{x' \in \cX}\pi(x)^{1-q}\pi(x')^{1-q} 
\leq \frac{ \pi(z)^{1-q} }{ (1 - \N S^{-(1-q)})^2}.
\end{align*}
(See also Lemma 3 of \citet{zhou2022rapid}.)
Finally, set $q = \log(S/\N)/ (2\log S)$ so that $S^{-q} = \N S^{-(1 - q)}$; note that $q \in (0, 1)$ since $S > \N \geq 1$. 
We then obtain the asserted result by taking maximum over $e\in E_S$. 
\end{proof}

\section{Path methods for restricted spectral gaps}\label{sec:appx-spec-gap}
\subsection{Restricted spectral gap}\label{sec:appx-gap-atchade}
If there are isolated local modes, the spectral gap of the Markov chain can be extremely close to 0. This remains true even when all local modes (other than the global one) have negligible probability mass are highly unlikely to be visited by the chain if it is properly initialized. 
As a result, calculating the standard spectral gap bound yields an overly conservative estimate of the mixing time.   
To address this issue, we  work with the restricted spectral gap, which was considered in~\cite{atchade2021approximate}.  

We still use the notation introduced in Appendix~\ref{sec:appx-prelim}. 
Given a function $f \colon \cX \rightarrow \bbR$, let 
\begin{equation}\label{eq:def-Lp-norm-pi}
    ||f||_{L^m(\pi)} = \left( \sum_{x \in \cX} |f(x)|^m \pi(x) \right)^{1/m}, \text{ for } m \in (2, \infty].
\end{equation}  
Given $\cX_0 \subset \cX$, define the $\cX_0$-restricted spectral gap by 
\begin{equation}\label{eq:def-gap-restricted}
\begin{aligned}
   \mathrm{Gap}_{\cX_0}(\bP)   
=\;&  \inf_{f \colon \mathrm{Var}_\pi(f) > 0}  \frac{\sum_{x,y \in \cX_0} (f(x)-f(y))^2 \bP(x,y) \pi(x)}{\sum_{x,y \in \cX_0}(f(x)-f(y))^2 \pi(x) \pi(y) }. 
\end{aligned}
\end{equation}
We note that if $\cX_0 = \cX$, $\mathrm{Gap}_{\cX_0}(\bP) $ coincides with $\mathrm{Gap}(\bP)$. 
The following lemma is from \cite{atchade2021approximate},
which shows that we can use restricted spectral gap to bound the mixing time given a sufficiently good initial distribution.  

\begin{lemma}\label{lm:yves}
Let $\epsilon \in (0, 1/2)$ and $\pi_0$ be another distribution on $\cX$. 
Define $f_0 = \pi_0 / \pi$. 
Choose any $m \in (2, \infty]$, and let $B \in [1, \infty)$ be such that $||f_0 ||_{L^m(\pi)} \leq B$.  
If $\cX_0 \subset \cX$ satisfies 
\begin{align*}
    \pi(\cX_0) \geq 1 - \left( \frac{\epsilon^2}{5B^2} \right)^{1 + \frac{2}{m - 2} }, 
\end{align*}
then $|| \pi_0 P^t - \pi ||_{\rm{TV}} \leq \epsilon$ for any $t$ such that 
\begin{align*}
    t \geq \frac{ \log (B^2 / 2 \epsilon^2) }{ \mathrm{Gap}_{\cX_0} (\bP) }. 
\end{align*}
\end{lemma}
\begin{proof}
    This follows from Lemma 1 and Lemma 2 of~\cite{atchade2021approximate} and the argument after Lemma 2 in~\cite{atchade2021approximate}. 
    Note that \cite{atchade2021approximate} used a different scaling in the definition of the total variation distance. 
\end{proof}

\subsection{Multicommodity flow bounds}\label{sec:appx-flow-bound} 
Given a subset $\cX_0 \subseteq \cX$, recall that we write 
\begin{align}
    \cN|_{\cX_0}(x)  = \{ x' \in \cX_0 \colon \bP(x, x') > 0\}, 
\end{align}
for each $x \in \cX_0$.   
Similarly, let 
\begin{equation}\label{eq:def-E0}
    \barE_0 = \left\{ (x, x') \in \cX_0^2 \colon   x \neq x' \text{ and } \bP(x, x') > 0 \right\}
\end{equation}
denote the edge set of $\bP$ restricted to $\cX_0$.  
Given $E_0 \subset \barE_0$,  let $\Gamma_0(x, x') = \Gamma_{E_0}(x, x')$ and $\Gamma_0 = \Gamma_{E_0}$  denote the corresponding path sets. 
We provide the following proposition that extends the argument used in Proposition~\ref{prop:gap_bound} to the restricted spectral gap.

\begin{proposition}\label{prop:restricted_gap_bound}
Let  $E_0\subset \barE_0$, $w \colon E_0 \rightarrow (0, \infty)$, and for $\gamma \in \Gamma_0$, define $|\gamma|_w$ as in~\eqref{eq:def-length}. 
Let $\phi \colon \Gamma_{0} \rightarrow [0, \infty)$ satisfy that 
\begin{equation}\label{eq:sum-flow-property-restricted}
    \sum_{\gamma \in \Gamma_{0}(x, x')} \phi(\gamma) = \pi(x) \pi(x'), \; \forall x, x' \in \cX_0,\;  x \neq x'. 
\end{equation} 
Then $\mathrm{Gap}_{\cX_0}(\bP) \geq 1/A(E_0, w, \phi)$, where $A$ is defined by~\eqref{eq:def-congestion-A}.
\end{proposition}

\begin{proof}
By Lemma~\ref{lemma:gap_bound}, for any $f \colon \cX \rightarrow \bbR$,  
\begin{align*}
      \sum_{x,x' \in \cX} |f(x') - f(x)|^2 \Phi(x, x')  
    \leq \;&  A(E_0, w, \phi) \sum_{(x, x') \in E_0}  Q( (x, x') )   |f(x') - f(x)|^2, 
\end{align*} 
where  $\Phi(x, x') = \sum_{\gamma \in \Gamma_E(x, x')} \phi(\gamma)$.
If $x \notin \cX_0$ or $x' \notin \cX_0$, $\Gamma_0(x, x') = \emptyset$ and thus $\Phi(x,x') = 0$. Hence, 
\begin{align*}
     \sum_{x,x' \in \cX} |f(x) - f(x')|^2 \Phi(x, x')   
    =\;& \sum_{x, x' \in \cX_0} |f(x) - f(x')|^2 \Phi(x, x') \\
    =\;& \sum_{x, x' \in \cX_0} |f(x) - f(x')|^2 \pi(x)\pi(x').  
\end{align*}
Further, since $E_0 \subset \cX_0^2$, 
\begin{align*}
    \sum_{(x, x') \in E_0}  Q( (x, x') )   |f(x) - f(x')|^2  \leq 
     \sum_{x, x' \in \cX_0} |f(x) - f(x')|^2 \bP(x,x')\pi(x). 
\end{align*}
The result then follows from the definition of the restricted spectral gap $\mathrm{Gap}_{\cX_0}(\bP)$.  
\end{proof}

\subsection{Application of the multicommodity flow method
under the restricted unimodal condition} \label{sec:appx-flow-restricted-gap-unimodal} 

In this section, we assume $\R|_{\cX_0} > \N$ and show that the multicommodity flow method reviewed in Section~\ref{subsec:review}  is also applicable to bounding the restricted spectral gap. 
For $S \in (\N, \R|_{\cX_0}]$, we define
\begin{align}\label{eq:condition_restricted}
    \quad \cN|_{\cX_0}^S(x) = \left\{ x' \in \cN|_{\cX_0}(x): \pi(x') / \pi(x) \geq S \right\},
\end{align}
and 
\begin{equation}\label{eq:restricted_edge_set}
    E_S^0 = \{(x,x'): x \in \cN|^S_{\cX_0}(x') \text{ or } x' \in \cN|^S_{\cX_0}(x)\}.
\end{equation}
Let $\Gamma_S^0 = \Gamma_{E_S^0}$ denote the resulting path set. 

\begin{lemma}\label{lm:bounding_A0}
    Suppose the triple $(\cX_0, \cN|_{\cX_0}, \pi)$ satisfies  
    $$ \R|_{\cX_0} = \R (\cX_0, \cN|_{\cX_0}, \pi) > \N(\cX, \cN) = \N , $$ 
    and let $x^*_0 = \argmax_{x \in \cX_0} \pi(x)$. 
    Let $E_S^0$ be given by~\eqref{eq:restricted_edge_set}
    with $S \in (\N, \R|_{\cX_0}]$.    
    Then,  there exist $w \colon E_S^0 \rightarrow (0, \infty)$, 
    $\phi \colon \Gamma_S^0 \rightarrow [0, \infty)$ such that $\phi$ satisfies~\eqref{eq:sum-flow-property-restricted} and 
    \begin{equation}\label{eq:final_bound_restricted}
         A(E_S^0, w, \phi) \leq \frac{ c(S/M)}{2 }  \max_{z \in \cX_0 \setminus \{x^*_0\}} \frac{1}{\bP(z,\cN|_{\cX_0}^S(z))}, 
    \end{equation}
    where $A$ is given by~\eqref{eq:def-congestion-A} and 
    $c(u) = 4(1 - u^{-1/2})^{-3}.$
\end{lemma}
\begin{proof} 
Let $\bP|_{\cX_0}$ denote the restriction of $\bP$ to $\cX_0$; that is, $\bP|_{\cX_0}(x, x') = \bP(x, x')$ if $x, x' \in \cX_0$ and $x \neq x'$. 
Apply Lemma~\ref{lm:def_flow} to $\bP|_{\cX_0}$ on $(\cX_0, \cN|_{\cX_0})$, which is allowed since $\R|_{\cX_0} > 1$.  
We then obtain   a flow $\phi\colon \Gamma_S^0 \rightarrow [0, \infty)$ such that~\eqref{eq:sum-flow-property-restricted} is satisfied and  
\begin{align} 
\sum_{\substack{\gamma\colon \gamma \in \Gamma_S^0(x,x') \\ \textit{s.t. } (z,z') \in \gamma} }
    \phi(\gamma)  \leq \;& \frac{\bP|_{\cX_0}(z,z')}{\bP|_{\cX_0}(z, \,\cN|_{\cX_0}^S(z))} \pi(x) \pi(x') \\ 
     = \;& \frac{\bP(z,z')}{\bP(z, \,\cN|_{\cX_0}^S(z))} \pi(x) \pi(x')
\end{align} 
for any $x, x' \in \cX_0$ with $x \neq x'$ and  $ (z,z') \in E_S^0$ with $z' \in \cN|_{\cX_0}^S(z)$. 

The asserted bound on $A(E_S^0, w, \phi)$ can be derived by repeating the argument used for proving Lemma~\ref{lm:bounding_A} on the restricted space $(\cX_0, \cN|_{\cX_0})$. 
Note that the only difference is that $\pi$ is no longer a probability measure on $(\cX_0, \cN|_{\cX_0})$. But this does not affect the proof since the proof of Lemma~\ref{lm:bounding_A} only relies on $\max_x \pi(x) \leq 1$. 
\end{proof}

\section{Review of single-element drift condition}\label{sec:appx-drift}

The drift-and-minorization method  is one of the most frequently used techniques for deriving the convergence rates of Markov chains on general state spaces. 
It directly bounds the total variation distance from the stationary distribution without using the spectral gap. 
For our purpose, we only need to use the single-element drift condition considered in~\cite{jerison2016drift}, which is particularly useful on discrete spaces. 

\begin{definition}[single-element drift condition]\label{def:single}
    The transition matrix $\bP$  satisfies a \textit{single-element drift condition} if there exist an element $x^* \in \cX$, a function $V: \cX \rightarrow [1, \infty)$, and constants $0 <\alpha < 1$ such that $(\bP V) (x) \leq \alpha V (x)$ for all $x \in \cX \setminus \{x^*\}$. 
\end{definition}

Note that in~\cite{jerison2016drift}, the single-element drift condition   is proposed for general-state-space Markov chains, and one needs to further require that $(\bP V)(x^*) < \infty$ (this holds trivially in our case since $\cX$ is finite).  
Given this drift condition, we can bound the mixing time using the following result of~\cite{jerison2016drift}. 
Recall that we assume $\bP$ is reversible and thus its eigenvalues are all real. 

\begin{lemma}[Theorem 4.5 of~\cite{jerison2016drift}] 
\label{lm:single-element-drift}
Suppose $\bP$ has non-negative eigenvalues and satisfies the single-element drift condition given in Definition~\ref{def:single}.  
Then, for all $t \geq 0, x \in \mathcal{X}$,
    \begin{align*}
        \left\|\bP^t(x, \cdot)-\pi\right\|_{\mathrm{TV}} \leq 2 V(x) \alpha^{t+1}. 
    \end{align*}
Therefore, we have    
\begin{align*}
    \tmix_x(\bP, \epsilon) \leq \frac{ 1 }{ 1 - \lbda}\log \left( \frac{2V(x)}{\epsilon}\right). 
\end{align*}
\end{lemma}

\begin{proof}
This directly follows from Theorem 4.5 of~\cite{jerison2016drift} and the inequality $\log \alpha < \alpha -1$ for $\alpha \in (0, 1)$
\end{proof}

\section{Proofs for general mixing time bounds}\label{sec:appx-proof}

\begin{proof}[Proof of Lemma~\ref{lm:accept-prob}] \label{proof:lm-accept-prob}
Since $Z_h(x)\geq  L$, we have 
\begin{align*}
\bK_h\left(x, x'\right) &= \frac{ h (\pi(x') / \pi(x))}{ Z_h(x) }  \leq  \frac{h(\pi(x')/\pi(x))}{ L }, \\
\bK_h\left(x', x\right) &= \frac{ h (\pi(x) / \pi(x'))}{ Z_h(x') } \geq  \frac{ \ell }{ \N L }, 
\end{align*}
where the second equality follows from $h \in [\ell, L]$ and $|\cN(x')| \leq \N$. 
Combining the two inequalities, we get
\begin{align*}
    \frac{\pi\left(x'\right) \bK_h \left(x', x\right)}{\pi(x) \bK_h \left(x, x'\right)} & \geq \frac{\pi(x')/\pi(x) }{h(\pi(x')/\pi(x))} \frac{\ell}{\N}. 
\end{align*}
Since $\pi(x')/\pi(x) \geq \ell$ and $h(u) \leq u$ for any $u \geq \ell$, we have 
\begin{align*}
    \frac{\pi\left(x'\right) \bK_h \left(x', x\right)}{\pi(x) \bK_h \left(x, x'\right)} & \geq  \frac{\ell}{\N} \geq 1, 
\end{align*}  
where the last inequality  follows by assumption. 
\end{proof}

\begin{proof}[Proof of Theorem~\ref{th:mix-imh}]\label{proof:th-mix-imh}
Applying Lemma~\ref{lm:bounding_A} with  $S =L/\N$, we get 
\begin{align*}
A(E_S, w,\phi) \leq \frac{c(\rho)}{2} \left\{\min_{x \in \cX \setminus \{x^*\}}\bP_h^\mathrm{lazy}(x, \cN_S(x))\right\}^{-1},
\end{align*}
where $E_S$ is defined by~\eqref{eq:def-ES}. 
We prove in Lemma~\ref{lm:proposal-prob} below that for $x \neq x^*$, 
$\bP_h(x, \cN_S(x))  \geq 1/2$, and thus
$\bP_h^{\mathrm{lazy}}(x, \cN_S(x))  \geq 1/4$. 
We conclude the proof by applying Proposition~\ref{prop:gap_bound} and inequality~\eqref{eq:ineq-gap-mix}.
\end{proof}

\begin{lemma}\label{lm:proposal-prob}
Let $S = L/M$. Under the setting of Theorem~\ref{th:mix-imh}, we have $\bP_h(x, \cN_S(x))  \geq 1/2$ for each $x \neq x^*$,
where $\cN_S(x)$ is defined in~\eqref{eq:condition1}. 
\end{lemma}
\begin{proof} 
Fix any $x \neq x^*$. By the definition of $\R$, there exists some $x' \in \cN(x)$ such that $\pi(x')/\pi(x) \geq \R \geq L$. Hence, $Z_h(x) \geq L$.  Using $S = L / \N > \ell = \N$ and Lemma~\ref{lm:accept-prob}, we see that the acceptance probability for every $x' \in \cN_S(x)$ equals one.  
Therefore, it only remains to prove that  $\bK_h(x, \cN_S(x)) \geq 1/2$.  
On one hand,  $\R \geq L$ and $x \neq x^*$ implies that 
\begin{align*}
Z^S_h(x) \coloneqq \sum_{x' \in\, \cN_S(x)} h\left( \frac{\pi(x')}{\pi(x)} \right)    \geq L. 
\end{align*}
On the other hand, since $|\cN(x)| \leq \N$, we have 
\begin{align*}
Z_h(x) - Z^S_h(x)  =  \sum_{x' \in\, \cN(x) \setminus \cN_S(x)} h\left( \frac{\pi(x')}{\pi(x)} \right) \leq \N \ell = \N^2. 
\end{align*}
Since $L > \N^2$,  $\bK_h(x, \cN_S(x)) = Z^S_h(x) / Z_h(x) \geq 1/2$, which completes the proof.
\end{proof}

\begin{proof}[Proof of Theorem~\ref{th:mix-imh-drift}]
For $x \neq x^*$,
\begin{align*}
    (\bP_h V)(x) = \bP_h(x,x) V(x) +\sum_{x'\in\, \cN(x)} \bP_h(x,x') V(x').  
\end{align*}
Using $\bP_h(x, x) = 1 - \sum_{z\in\, \cN(x)} \bP_h(x,x')$, we can rewrite the above equation as 
\begin{equation}\label{eq:PV}
    \frac{(\bP_h V)(x)}{V(x)} = 1 +\sum_{x'\in\, \cN(x)} \bP_h(x,x') D(x, x'),  
\end{equation}
where we define
\begin{equation}\label{eq:def-delta}
    D(x, x) = \frac{V(x')}{V(x)} - 1
    = \exp\left( \frac{\log \frac{\pi(x')}{\pi(x)} }{ \log \pimin} \right)  - 1.
\end{equation}
Since $D(x, x') \leq 0$ whenever $\pi(x') \geq \pi(x)$, we have 
\begin{align}
   & \frac{(\bP_h V)(x)}{V(x)} \leq 1 + A_1 + A_2,  \label{eq:drift-decomp}\\ 
 \text{ where }  A_1 \coloneqq \sum_{x'\in \,\cN_S(x)} \bP_h(x,x')& D(x, x'), \quad A_2 \coloneqq \sum_{ \substack{x': x'\in \,\cN(x) \\ \textit{s.t. } \pi(x') \leq \pi(x)}}  \bP_h(x,x') D(x, x'), 
\end{align}
with $S = L/M$ and $\cN_S$ is given by~\eqref{eq:condition1}. We now consider the two cases separately. 

\medskip
\noindent \textit{Case 1: $x' \in \cN_S(x)$.} 
Using $e^{-a} - 1 \leq -a/2$ for $a \in [0, 1]$ and the definition of $\cN_S(x)$, we obtain that 
\begin{align*}
    D(x, x') \leq \frac{\log (\pi(x') / \pi(x)) }{2  \log \pimin}
    \leq \frac{\log (L/\N) }{2  \log \pimin}. 
\end{align*}
By Lemma~\ref{lm:proposal-prob}, $\bP_h(x, \cN_S(x)) \geq 1/2$. Hence,
\begin{equation}\label{eq:A1}
    A_1  \leq \frac{\log (L/\N) }{4  \log \pimin}. 
\end{equation}

\medskip
\noindent \textit{Case 2: $\pi(x') \leq \pi(x)$.} For $D(x, x')$, we simply use the worst-case bound 
$D(x, x') \leq e - 1$. 
Since $\pi(x') \leq \pi(x)$ and $\ell = \N$, we have $h( \pi(x')/\pi(x) ) = \N$.
It follows that 
\begin{equation}
    \bP_h(x, x') \leq \bK_h(x, x') = \frac{ \N }{Z_h(x)} \leq \frac{\N}{L}, 
\end{equation}
where the last inequality follows from the assumption $L \leq \R$. 
Since $|\cN(x)| \leq \N$, we get 
\begin{equation}\label{eq:A2}
    A_2 \leq \frac{\N^2}{L} (e - 1). 
\end{equation}

\medskip 
Combining~\eqref{eq:drift-decomp},~\eqref{eq:A1} and~\eqref{eq:A2}, we  get 
\begin{align}
1 - \lbda \geq  - \frac{\log (L/\N) }{4  \log \pimin}  - \frac{\N^2}{L} (e - 1). 
\end{align}
By Lemma~\ref{lm:single-element-drift},  
\begin{align}
    \tmix(\bP_h, \epsilon) \leq \frac{ 1 }{ 1 - \lbda}\log \left( \frac{2 e}{\epsilon}\right), 
\end{align}
which yields the asserted result. 
\end{proof}

\begin{proof}[Proof of Theorem~\ref{th:mix-rwmh-approx}]\label{proof:mix-rwmh-approx}
Let $\delta_{x_0}$ be the Dirac measure which assigns unit probability mass to some $x_0 \in \cX$, and define $f_0 = \delta_{x_0} / \pi$.  Hence, we can write $f_0(x)=  \ind_{\{x = x_0\}}(x)/\pi(x)$. 
Applying Lemma~\ref{lm:yves} with $m = \infty$, we obtain that 
\begin{equation}\label{eq:approx-gap-mix}
    \tmix_{x_0}(\bP_0^\mathrm{lazy}, \epsilon) \leq \frac{ 1 }{ \mathrm{Gap}_{\cX_0} (\bP_0^\mathrm{lazy}) } \log \left\{\frac{1}{2 \epsilon^2 \, \eta^2 } \right\},
\end{equation}
if $\pi(\cX_0) \geq 1 - \epsilon^2 \eta^2 / 5$.
Now, we show that the restricted spectral gap is bounded as
\begin{align*}
\mathrm{Gap}_{\cX_0} (\bP_0^\mathrm{lazy}) \geq \frac{1}{Mc(\rho)}.
\end{align*}
We use Lemma~\ref{lm:bounding_A} with the restricted edge set $E_S^0$ defined in~\eqref{eq:restricted_edge_set} and Proposition~\ref{prop:restricted_gap_bound}.
By setting $S = \R|_{\cX_0}$ in~\eqref{eq:final_bound_restricted}, we obtain
\begin{align*}
    A(E_S^0, w,\phi) \leq \frac{c(\rho)}{2} \left\{\min_{x \in \cX_0 \setminus \{x^*_0\}}\bP^\mathrm{lazy}_0(x, \cN|_{\cX_0}^S(x))\right\}^{-1}.
\end{align*}
For any $x \in \cX_0 \setminus \{x^*_0\}$, define 
 $$\hat{x} = \arg \max_{\,x' \in \cN|_{\cX_0}^S(x)}\pi(x').$$ 
The definition of  $\R|_{\cX_0} $ implies that $\pi(\hat{x})/ \pi(x) \geq  \R|_{\cX_0} \geq \N$. Hence,  
$$\bP^\mathrm{lazy}_0(x, \cN|_{\cX_0}^S(x)) \geq \bP^\mathrm{lazy}_0(x, \hat{x})  \geq 1/(2M).$$ 
Application of Proposition~\ref{prop:restricted_gap_bound} yields the conclusion. 
\end{proof}

\begin{proof}[Proof of Theorem~\ref{th:mix-imh-approx}]
As in the proof of Theorem~\ref{th:mix-rwmh-approx}, we use Lemma~\ref{lm:yves} to get  
\begin{align*}
    \tmix_{x_0}(\bP_h^\mathrm{lazy}, \epsilon) \leq \frac{ 1 }{ \mathrm{Gap}_{\cX_0} (\bP_h^\mathrm{lazy}) } \log \left\{\frac{1}{2 \epsilon^2 \, \eta^2 } \right\},
\end{align*}
if $\pi(\cX_0) \geq 1 - \epsilon^2 \eta^2 / 5$. 
We use Lemma~\ref{lm:bounding_A} with $S = L/\N$, and the bound given in~\eqref{eq:final_bound_restricted} becomes 
\begin{align*}
A(E_S^0, w,\phi) \leq \frac{c(\Tilde{\rho})}{2} \cdot  \left\{\min_{x \in \cX_0 \setminus \{x^*_0\}}\bP^\mathrm{lazy}_h(x, \cN|_{\cX_0}^S(x))\right\}^{-1}.
\end{align*}
As in the proof of Lemma~\ref{lm:proposal-prob}, we can use Lemma~\ref{lm:accept-prob} to show that any proposal from $x$ to $x' \in \cN|_{\cX_0}^S(x)$ has acceptance probability equal to $1$. Hence, 
\begin{equation}
    \bP_h^\mathrm{lazy}(x, \cN|_{\cX_0}^S(x)) = \bK_h^\mathrm{lazy}(x, \cN|_{\cX_0}^S(x)) = \frac{\bK_h^\mathrm{lazy}(x, \cN|_{\cX_0}^S(x)) }{\bK_h^\mathrm{lazy}(x, \cX_0 ) } \bK_h^\mathrm{lazy}(x, \cX_0 ) \geq \frac{1}{4} \bK_h^\mathrm{lazy}(x, \cX_0 ),
\end{equation}
where the last inequality follows from Lemma~\ref{lm:proposal-prob} (note that an additional factor of $1/2$ arises from the laziness of the chain).   
By the condition given in~\eqref{eq:mix-imh-approx-condition}, $\bP_h^\mathrm{lazy}(x, \cN|_{\cX_0}^S(x)) \geq 1/8$. 
The conclusion then follows from Proposition~\ref{prop:restricted_gap_bound}.  
\end{proof}

\section{A slow mixing example with $\R < \N$}\label{sec:examples} 

Let $[p] = \{1, 2, \dots, p\}$. Define the state space by 
\begin{equation*}
    \mathcal{X} =
\bigcup\nolimits_{k = 0}^{ 2^p } [p]^k,  
\end{equation*}
where we set $[p]^0 = \emptyset$.  
The size of the state space is $$|\mathcal{X}| =  \sum_{k=0}^{  2^p  } p^k  = \frac{ p^{  2^p  +1} - 1}{p-1}. $$
For $x = (x_1, \dots, x_k) \in [p]^k$ with $0 < k <2^p$, its neighborhood is defined by 
\begin{equation*}
   \cN(x) = \left(\bigcup_{\ell = 1}^{p} \{(x_1, \dots, x_{k}, \ell)\} \right) \cup \{(x_1, \dots, x_{k-1})\}; 
\end{equation*}
In other words, $\cN(x)$ consists of  all states  obtained by either appending an element in $[p]$ to $x$ or removing the last element $x_k$. 
When $k = 2^p$, define $\cN(x) = \{(x_1, \dots, x_{k - 1})\}$, and when $k = 0$, define $\cN(\emptyset) = [p]$. 
Clearly, we have $M = \max |\cN(x)| = p + 1$.  
Let the target distribution be given by 
\begin{equation*}
    \pi(x) \propto  p^{-k}, \text{ if } x \in [p]^k,  
\end{equation*} 
which yields $x^* = \emptyset$ and $R = p$; thus $R$ is slightly less than $M$. 
Define $\mathcal{A}_j = \bigcup_{k = j}^{2^p} [p]^k$. It is easy to show that  
\begin{align*}
    \pi(\mathcal{A}_j) = \frac{ 2^p - j + 1 }{2^p + 1 }. 
\end{align*} 
Now suppose that the Markov chain  starts at $x^* = \emptyset$. 
Since $\mathbf{P}^t(\emptyset, \mathcal{A}_{t + 1} ) = 0$,  in order that  $||\mathbf{P}^t(0, \cdot) - \pi ||_{\mathrm{TV}}  \leq 1/4$ we need $\pi( \mathcal{A}_{t + 1} ) \leq 1 / 4$. But this condition implies that 
\begin{equation*}
    \frac{ 2^p - (t + 1) + 1 }{2^p + 1 } \leq \frac{1}{4} \; \Longrightarrow \;  t \geq \frac{3}{4}(2 ^p + 1) - 1.  
\end{equation*}
This proves that the mixing time grows exponentially  in $p$.

\section{Corollaries for variable selection}\label{sec:cor-settings-vs}

\subsection{Model and assumptions}\label{subsec:cor-vs}  
We consider the Bayesian variable selection model studied in~\cite{yang2016computational}, which assumes  
\begin{align*} 
    \by  = \bX_\delta \bbeta_\delta + \bz,  \quad \bz \sim \mathrm{MVN}(0, \phi^{-1} I_n),
\end{align*} 
and uses the following prior distribution on $(\bbeta, \phi, \delta)$: 
\begin{equation}\label{eq:prior}
    \begin{aligned}
    & \bbeta_\delta  \mid \delta \sim \mathrm{MVN}(0, g\phi^{-1} (\bX_\delta^{\T}\bX_\delta)^{-1} ),\\
    & \pi_0(\phi)  \propto \phi^{-1}, \\
    & \pi_0(\delta)  \propto p^{-\kappa ||\delta||_1}.
    \end{aligned}
\end{equation}
By integrating out the model parameters $(\bbeta_\delta, \sigma^2)$, we obtain the marginal posterior distribution $\pi(\delta)$ given in~\eqref{eq:post}:  
\begin{align*}
\quad \pi(\delta) \propto \frac{1}{p^{\kappa ||\delta||_1}} \cdot \frac{(1+g)^{-(||\delta||_1 / 2)}}{\left\{1+g\left( 1 - \mathsf{r}^2(\delta) \right)\right\}^{n / 2}},  
\end{align*}
where we recall $\mathsf{r}^2(\delta)= \by^\T (\bX_\delta(\bX_\delta^\T \bX_\delta)^{-1} \bX_\delta^\T) \by / (\by^\T \by)$. 
We still use $s$ to denote the maximum model size on the restricted search space $\cV_s$.  
Let $\cV, \cV_s$ be equipped with the neighborhood relation $\cN_{\rm{ads}}$ and consider the corresponding MH algorithms, which are also known as the add-delete-swap samplers.

Let the true data-generating model be 
\begin{equation}
    \by \sim \mathrm{MVN}( \bX \bbeta^*,  \sigma_0^2 \bI ), 
\end{equation}
where $\bbeta^*$ is the true regression coefficient vector and $\sigma_0^2$ is the true error variance.   
Consider a high-dimensional setting where $n, p, s$ all tend to infinity. 
Below are the assumptions made in~\cite{yang2016computational}  about the true data-generating process, design matrix and prior parameters.  
Note that  to simplify the mixing time bounds to be derived, we make one modification by letting the parameter $\alpha$ in (A3)  be fixed. 
Assumption (A5) is the beta-min condition used in Theorem~2 of~\cite{yang2016computational}.

\begin{enumerate}[label=(A\arabic*)]
    \item Let $S = \{ j \in [p] \colon |\bbeta^*_j| \geq C_{\bbeta} \}$ be the ``influential'' variables, where $C_{\bbeta} > 0$ denotes the minimal signal size.  
Then $\bbeta^*$ satisfies  
\begin{align*}
 \frac{1}{\sqrt{n}} \|\bX \bbeta^*\|_2^2 \leq g \sigma_0^2 \frac{\log p}{n}, 
 \quad \frac{1}{\sqrt{n}} \|\bX_{S^c} \bbeta^*_{S^c}\|_2^2 \leq \tilde{\mathcal{L}} \sigma_0^2 \frac{\log p}{n},  
\end{align*}
for some universal constant \( \tilde{\mathcal{L}} > 0 \).
\item    The design matrix has normalized columns, that is, $\|\bX_j\|_2^2 = n$ for  $j \in [p]$. Letting \( Z \sim \mathrm{MVN}(0, I_n) \), there exist constants \( \nu = \nu(n) > 0 \) and \( \mathcal{L} = \mathcal{L}(n) < \infty \) such that
\begin{align*}
    \min_{||\delta||_1 \leq s} \lambda_{\min} \left( \frac{1}{n} \bX_{\delta}^T \bX_{\delta} \right) &\geq \nu,  \\ 
     \mathbb{E}_Z \left[ \max_{||\delta||_1 \leq s} \max_{k \in [p] \setminus \delta} \frac{1}{\sqrt{n}} 
\left| \langle (I - \Phi_{\delta}) \bX_k, Z \rangle \right| \right] &\leq \frac{1}{2} \sqrt{\mathcal{L} \nu \log p}.
\end{align*}
\item The prior parameters satisfy that $\kappa \geq 2$,  $g \asymp p^{2\alpha}$ for some universal constant $\alpha \geq 1/2$, and 
\begin{equation}\label{eq:vs-prior-cond}
 \kappa + \alpha \geq 4(\mathcal{L} + \tilde{\mathcal{L}} ) + 2.   
\end{equation} 
\item      The sparsity parameter \( s \) satisfies 
\[ 
\max \left\{ 1, (2 \nu^{-2} \omega(\bX) + 1) s^* \right\} \leq s \leq \frac{1}{32} \left\{ \frac{n}{\log p} - 8 \tilde{\mathcal{L}}   \right\}, 
\]
where $s^* = |S|$ is the true sparsity and $\omega(\bX) = \max_{\delta \in  \cV_s } \| (\bX_{\delta}^T \bX_{\delta})^{-1} \bX_{\delta}^T \bX_{\delta^* \setminus \delta} \|_{\text{op}}^2.$  
\item The parameter $C_{\bbeta}$  defined in (A1) satisfies 
\begin{equation*} 
C_{\bbeta}^2  \geq \frac{ 
128 (\alpha + \kappa + \mathcal{L}) \sigma_0^2 \log p}{\nu^2 n }. 
\end{equation*} 
\end{enumerate}

To prove mixing time results for informed MH samplers, we consider a stronger version of Assumption (A3): 
\begin{enumerate}[label=(A\arabic*'), start=3]
\item  The prior parameters satisfy that $\kappa \geq 2$,  $g \asymp p^{2\alpha}$ for some universal constant $\alpha \geq 1/2$, and  
\begin{equation}
    \kappa + \alpha \geq 4 (\mathcal{L}+ \tilde{\mathcal{L}} + 1)  +  2 \xi \text{ for some universal constant } \xi > 0. 
\end{equation} 
\end{enumerate}

\subsection{Proofs}\label{sec:proof-vs}

\begin{proof}[Proof of Corollary~\ref{cor:vs_thm1}]
This is the setting of Theorem 2 in~\cite{yang2016computational}, where the authors obtain a mixing time upper bound of the form
\begin{equation}\label{eq:tmix-bound-YWJ}
    \tau(\mathsf{P}_0^{\text{lazy}}, \epsilon) \leq 12p s^2 \left\{ (\alpha n + \alpha s + 2 \kappa s) \log p + \log \epsilon^{-1} + \log 2 \right\}. 
\end{equation}
We show that Theorem~\ref{th:rwmh-1} can be used to improve this bound. 
Lemma~4 in \cite{yang2016computational} proves that with probability $1 - O(p^{-1})$, we have \( R > p^2 \), where $R$ is the parameter in our unimodal condition. Since the maximum neighborhood size is bounded by $M = p(s + 1)   < R$, by Theorem~\ref{th:rwmh-1}, 
\begin{align*}
    \tau(\mathsf{P}_0^{\text{lazy}}, \epsilon) = O_p \left( M \log \frac{1}{\epsilon \pimin } \right) = O_p \left( p s \log \frac{1}{\epsilon \pimin } \right), \text{ where } \pimin = \min_{\delta \in \cV_s} \pi_n(\delta). 
\end{align*}
In the proof of Theorem 2 in~\cite{yang2016computational}, it is shown that 
\begin{equation}\label{eq:vs-pmin}
    \pimin \geq  \frac{1}{2} p^{-(\kappa + \alpha/2)s} p^{-\alpha n /2}. 
\end{equation}
Hence, treating $\epsilon$ as a fixed constant 
\begin{equation*}
        \tau(\mathsf{P}_0^{\text{lazy}}, \epsilon)   = O_p \left( p s   ( \alpha n + \alpha s + 2 \kappa s )\log p  \right), 
\end{equation*}
which improves~\eqref{eq:tmix-bound-YWJ} by a factor of $s$.  The asserted bound follows by omitting the dependence on $\alpha$. 
\end{proof}

\begin{proof}[Proof of Corollary~\ref{cor:vs_thm3}]
By assuming $\alpha + \kappa \geq 4 (\mathcal{L}+ \tilde{\mathcal{L}} + 1) + 2 \xi$ in Assumption (A3),   the conclusion of  Lemma~4 in \cite{yang2016computational} can be strengthened to  $R > p^{4 + 2\xi}$. The proof is essentially the same, and a sketch of proof is provided in  Section S3 in~\cite{zhou2022dimension}. 
Since $L = \ell^{2 + \xi} \leq p^{4 +2 \xi}$, the conditions of Theorem~\ref{th:mix-imh} are satisfied. Hence, applying Theorem~\ref{th:mix-imh} we get 
\begin{align*}
    \tau(\mathsf{P}_h^{\text{lazy}}, \epsilon) = O_p \left(  \log \frac{1}{\epsilon \pimin } \right) = O_p \left(  \log \frac{1}{\epsilon \pimin } \right)  = O_p \left(    (  n +   \kappa s ) \log p  \right), 
\end{align*} 
where we have used~\eqref{eq:vs-pmin}.  

To apply Theorem~\ref{th:mix-imh-drift}, it remains to verify the additional condition~\eqref{eq:drift_pimin_condition}. By~\eqref{eq:vs-pmin}, this condition holds if $n = o( (ps)^\xi )$. The resulting mixing time bound is 
\begin{align*}
    \tau(\mathsf{P}_h^{\text{lazy}}, \epsilon) = O_p \left(  \frac{1}{\log (L / M)} \log \frac{1}{\epsilon \pimin } \right)    = O_p \left(  \frac{  n +   \kappa s }{1 + \xi}    \right). 
\end{align*}  
Omitting the dependence on $\xi$, we obtain the asserted bound. 
\end{proof}

\begin{proof}[Proof of Corollary~\ref{cor:vs_thm4}] 
We apply Theorem~\ref{th:mix-rwmh-approx} with $\cX = \cV$ and $\cX_0 = \cV_s$. Since the bound on $R|_{\cV_s}$ has already been addressed in the proof  of Corollary 1,  it only remains to bound the posterior mass on $\cV \setminus \cV_s$ by 
\begin{align*}
    \pi_n(\cV \setminus \cV_s  ) = \frac{
\sum_{\delta \in \cV \setminus \cV_s}
\int_{\theta, \phi} \frac{d\mathbb{P}_{\bbeta, \phi, \delta}}{d\mathbb{P}_0}(\by) \,
\pi_0(d\theta, d\phi, \delta)
}{
\sum_{\delta \in \{0,1\}^p}
\int_{\theta, \phi} \frac{d\mathbb{P}_{\bbeta, \phi, \delta}}{d\mathbb{P}_0}(\by) \,
\pi_0(d\theta, d\phi, \delta)
} \leq \frac{\epsilon^2\eta^2}{5}, 
\end{align*}
where $\mathbb{P}_{\bbeta, \phi, \delta}$ denotes the probability measure where $\by$ is generated from the linear regression model $\delta$ with regression coefficient vector $\bbeta$ and error variance $\phi^{-1}$, 
$\mathbb{P}_0$ denotes the true data-generating model with regression coefficient vector $\bbeta^*$ and error variance $\sigma_0^2$, and $\pi_0$ denotes the prior distribution.  
To this end, we calculate  an upper bound for the numerator, and a lower bound for the denominator by mimicking the proof arguments of Theorem 1 in~\cite{yang2016computational}.

For the numerator,  we use the fact 
\begin{align*}
 & \mathbb{E}_0 \left[ \sum_{\delta \in \cV \setminus \cV_s} \int_{\theta, \phi} \frac{d\mathbb{P}_{\bbeta, \phi, \delta}}{d\mathbb{P}_0}(\by) \,
\pi_0(d\theta, d\phi, \delta)\right] 
 = \pi_0(\cV \setminus \cV_s )   \\
& = \frac{1}{(1 + p^{-\kappa})^p} 
\sum_{t =  s+1 }^{p} \binom{p}{t} p^{-\kappa t} 
\leq \sum_{t =  s+1 }^{p}  p^{- (\kappa-1) t} \leq 2p^{-(\kappa - 1)(s+1)}, 
\end{align*}
where we have used $\kappa \geq 2$ to ensure the convergence of the geometric series. Hence, by Markov's inequality, 
\begin{align*}
    \mathbb{P}_0 \left[
\sum_{\delta \in \cV \setminus \cV_s }
\int_{\theta, \phi}
\frac{d\mathbb{P}_{\bbeta, \phi, \delta}}{d\mathbb{P}_0}(\by) \,
\pi_0(d\theta, d\phi, \delta)
\leq p^{-(\kappa-1)(s+1) + 1}
\right] \geq 1- O(p^{-1}).
\end{align*}

For the denominator, we can use the bound given in equation (39) in~\cite{yang2016computational}, which states that with high probability, 
\begin{align*}
    \sum_{\delta \in \{0,1\}^p}
\int_{\theta, \phi} \frac{d\mathbb{P}_{\bbeta, \phi, \delta}}{d\mathbb{P}_0}(\by) \,
\pi_0(d\theta, d\phi, \delta) \geq \int_{\theta, \phi}
\frac{d\mathbb{P}_{\bbeta, \phi, \delta^*}}{d\mathbb{P}_0}(\by) \,
\pi_0(d\theta, d\phi, \delta^*) \geq C p^{- (\kappa + \alpha + \tilde{\mathcal{L}} + 1)} 
\end{align*}
for some universal constant $C > 0 $. Since $\alpha, \tilde{\mathcal{L}}$ are universal constants, for sufficiently large $s$,  
\begin{align*}
    \pi(\cV \setminus \cV_s) =  O(p^{-(\kappa - 1)s + \alpha + \tilde{\mathcal{L}} +3}) \leq p^{-\kappa s / 2} 
\end{align*} 
holds with high probability. 

Apply Theorem~\ref{th:mix-rwmh-approx} with $\eta^2 = 5 p^{-\kappa s / 2} / \epsilon^2$ to conclude the proof. 
\end{proof}

\begin{proof}[Proof of Corollary~\ref{cor:vs_thm5}]  
The proof is analogous to that of Corollary 5. Apply Theorem~\ref{th:mix-imh-approx} with $\cX = \cV$ and $\cX_0 = \cV_s$.  
\end{proof}

\section{Corollaries for structure learning} \label{sec:cor-settings-sl}
 
\subsection{Model and assumptions}\label{subsec:cor-sl}

We first describe the notation and  setting for the Bayesian structure learning problem considered in~\cite{zhou2021complexity}.
Let \( X_{(i)} \) denote the \( i \)-th row of the data matrix \( X \). We model the conditional distribution of \( X \) given a DAG \( G \) and its Markov equivalence class (MEC) \( \mathcal{E}\) by
\begin{align*}
    X_{(1)}, \dots, X_{(n)}    \mid B, \Omega &\overset{\text{i.i.d.}}{\sim} \mathrm{MVN}(0, \Sigma(B, \Omega)), \\
    \Sigma(B, \Omega) &= (I - B^{\top})^{-1} \Omega (I - B)^{-1}, \\
    (B, \Omega) \mid G &\sim \pi_0(B, \Omega \mid G), \quad \forall (B, \Omega) \in \mathcal{D}_p(G), \\
    \pi_0(\mathcal{E}, G) & \propto   p^{- \kappa |G|} \pi_0(G \mid \mathcal{E})  \mathbf{1}_{\mathcal{C}_p(d_{\text{in}}, d_{\text{out}})}(\mathcal{E}), 
\end{align*} 
where $\kappa > 0$  is a hyperparameter, $|G|$ denotes the number of edges in $G$, 
\begin{align*}
    \mathcal{D}_p(G) &= \{ (B, \Omega) : B \in \mathbb{R}^{p \times p}, B_{ij} = 0 \text{ whenever the edge } i \to j \text{ is not in } G, \\
    &\Omega = \text{diag}(\omega_1, \dots, \omega_p) \text{ where } \omega_i > 0 \text{ for any } i \in [p] \}, 
\end{align*} 
and  $\mathcal{C}_p(d_{\text{in}}, d_{\text{out}})$ is the space of ``sparse MECs''. An MEC $\mathcal{E}$ is in $\mathcal{C}_p(d_{\text{in}}, d_{\text{out}})$ if and only if $\mathcal{E}$ contains a member DAG $G \in \mathcal{G}_p(d_{\text{in}}, d_{\text{out}})$, where   $\mathcal{G}_p(d_{\text{in}}, d_{\text{out}})$ is the space of all \( p \)-vertex DAGs with node in-degree bounded by $d_{\text{in}}$ and node out-degree bounded by $ d_{\text{out}}$.  We use $\mathcal{G}_p, \mathcal{C}_p$ to denote the corresponding unrestricted spaces. 
The prior distribution \( \pi_0(G \mid \mathcal{E}) \) satisfies that $\pi_0(G \mid \mathcal{E}) = 0$ if $G$ is not a member DAG  in $\mathcal{E}$. In~\citet{zhou2021complexity}, the term $p^{- \kappa |G|} $ in $\pi_0(\mathcal{E}, G)$ is replaced with $(c_1 p^{c_2})^{- |G|} $; we  omit $c_1$ and change the notation to align with that used in the variable selection problem. 
The conditional distribution of \( X \) given \( (B, \Omega) \) can also be expressed by the structural equation model (SEM),
\[
X_j = \sum_{i \neq j} B_{ij} X_i + \varepsilon_j, \quad \varepsilon_j \sim \mathrm{MVN}(0, \omega_j I),
\]
where \( \bbeta_j(G) \) is the subvector of the \( j \)-th column of \( B \) with entries indexed by  \( \text{Pa}_j(G) \), the parent set of node $j$;  the child set will be denoted by \( \text{Ch}_j(G) \). 
Let \( \pi_0(B, \Omega \mid G) \) be the empirical Bayes prior given by
\begin{align*}
& \pi_0(B, \Omega \mid \text{Pa}_1(G) = S_1, \dots, \text{Pa}_p(G) = S_p) 
\propto \\ \prod_{j=1}^{p} & \omega_j^{- \tilde{\kappa} /2 - 1} 
\mathrm{MVN}  \Bigg( \bbeta_j(G); (X_{S_j}^{\top} X_{S_j})^{-1} X_{S_j}^{\top} X_j, 
\frac{\omega_j}{\gamma} (X_{S_j}^{\top} X_{S_j})^{-1} \Bigg), 
\end{align*}
where $\gamma, \tilde{\kappa} > 0$ are hyperparameters.\footnote{
In~\citet{zhou2021complexity}, $\tilde{\kappa}$ is denoted by $\kappa$. 
}
By using the conjugacy of normal-inverse-gamma prior and  marginalizing out $G$ using the $\alpha$-fractional likelihood for some $\alpha \in (0, 1)$, we obtain the posterior distribution on $\mathcal{E}$: 
\[
\pi_n(\mathcal{E}) \propto e^{\psi(\mathcal{E})} \mathbf{1}_{\mathcal{C}_p(d_{\text{in}}, d_{\text{out}})}(\mathcal{E}),
\]
where $\psi(\mathcal{E}) = \sum_{j=1}^{p} \psi_j(\text{Pa}_j(G))$ for $ G \in \mathcal{E}$  with 
$$  e^{\psi_j(S)} =  p^{-\kappa |S|} (1 + \alpha \gamma^{-1})^{-|S|/2} \allowbreak
\left\{ X_j^{\top}  (I - X_S (X_S^{\top} X_S)^{-1} X_S^{\top}) X_j \right\}^{-(\alpha n + \tilde{\kappa})/2}.
$$

\noindent
\textbf{Neighborhood relation.} 
Define the neighborhood $\mathcal{N}_{\text{ads}}(G)$ on \( \mathcal{G}_p  \) as the union of three distinct neighborhoods corresponding to different types of edge modifications (addition, deletion, and swap operations). That is, $\mathcal{N}_{\text{ads}}(G) = \mathcal{N}_{\text{add}}(G) \cup \mathcal{N}_{\text{del}}(G) \cup \mathcal{N}_{\text{swap}}(G),
$ where  
\begin{align*}
    \mathcal{N}_{\text{add}}(G) &= \{ G' \in \mathcal{G}_p : G' = G \cup \{i \to j\} \text{ for some } i \to j \notin G \} ,\\
\mathcal{N}_{\text{del}}(G) &= \{ G' \in \mathcal{G}_p : G' = G \setminus \{i \to j\} \text{ for some } i \to j \in G \} ,\\
\mathcal{N}_{\text{swap}}(G) &= \{ G' \in \mathcal{G}_p : G' = (G \cup \{k \to j\}) \setminus \{\ell \to j\} \text{ for some } k \to j \notin G, \ell \to j \in G \}.
\end{align*} 
Define the neighborhood relation $\mathcal{N}_{\text{ads}}(\mathcal{E})$ on  \( \mathcal{C}_p  \) by  
$$\mathcal{N}_{\text{ads}}(\mathcal{E}) = \{  \mathcal{E}': 
\mathcal{E} \text{ has a member DAG } G \text{ and } \mathcal{E}' \text{ has a member DAG } G' \text{s.t.} G' \in \mathcal{N}_{\text{ads}}(G)  \}.$$ 
Our mixing time analysis considers MH algorithms on the restricted space  $( \mathcal{C}_p(d_{\text{in}}, d_{\text{out}}), \mathcal{N}_{\text{ads}}). $

\medskip 
\noindent
\textbf{True data-generating model.} Let $G^*$ and $\Sigma^* = \Sigma(B^*, \Omega^*)$ be the true DAG model and the true covariance matrix, respectively; assume $\Sigma^*$ is non-degenerate and $\mathrm{MVN}(0, \Sigma^*)$ is perfectly Markovian w.r.t.  $G^*$.      
We denote by $\bbS^p$ the set of all bijections from $[p]$ to $[p]$. An element $\sigma \in \bbS^p$ is said to be an ordering for a DAG $G$ if the following holds: for any $k < l$, the edge between the nodes $\sigma(k)$ and $\sigma(l)$, if it exists in $G$, is  directed as $\sigma(k) \rightarrow \sigma(l)$.   
Let $\mathcal{G}_p^\sigma = \{ G \in \mathcal{G}_p  \colon  \sigma \text{ is an ordering of } G   \}$. 
For any $\sigma \in \bbS^p$, define $(B_\sigma^*, \Omega_\sigma^*)$ to be the unique pair in $\mathcal{D}_p(\sigma) = \cup_{G \in \mathcal{G}_p^\sigma} \mathcal{D}_p(G)$  such that  
$$(I - (B_\sigma^*)^\top)^{-1} \Omega_\sigma^* (I - B_\sigma^*)^{-1} = \Sigma^*.$$ 
The DAG with weighted adjacency matrix $B_\sigma^*$ is called the minimal I-map of $G^*$ with respect to $\sigma$, which we denote by $G_\sigma^*$; that is, an edge $i \rightarrow j$ is in $ G_\sigma^*$ if and only if $(B_\sigma^*)_{ij} \neq 0$. 

\medskip 
\noindent
\textbf{Assumptions for high-dimensional analysis.} 
Below are the high-dimensional assumptions considered in Theorem 5 of~\citet{zhou2021complexity}, with slight modifications made to simplify the notation and the mixing time bounds to be derived. 
The model parameters are all implicitly indexed by $n$, and we assume $n, p, \din, \dout$ all tend to infinity. 

\begin{enumerate}[label=(B\arabic*)]
\item     There exist a constant $\nu > 1$  and a universal constant $\delta_0 > 0$ such that  
\begin{align*}
0 < \frac{1}{\nu (1 - \delta_0)^2} \leq \lmin(\Sigma^*) \leq \lmax(\Sigma^*) \leq  \frac{\nu}{(1 + \delta_0)^2}, 
\end{align*} 
where $\lmin, \lmax$ denote the smallest and largest eigenvalues, respectively.  
\item Define
\begin{equation}
d^*_{\mathrm{in}} = \max_{\sigma \in \bbS^p} \max_{j \in [p]} |\Pa_j (G^*_\sigma)  |, \quad 
    d^* = \max_{\sigma \in \bbS^p} \max_{j \in [p]} |\Pa_j (G^*_\sigma) \cup \Ch_j(G^*_\sigma)|.  
\end{equation}   
The sparsity parameters $\din$ and $\dout$ satisfy 
\begin{align*}
   & \max\left\{ \left(4   \nu^8 + 1\right)  d^*_{\mathrm{in}}, \;  d^*\right\}  \leq \din \leq \dout = o \left( \frac{n}{\log p} \right), \\ 
    & d^* \din + 1  \leq \dout \leq \frac{1}{2} t_0 \log_2 p, \text{ for some universal constant } t_0 >  0.  
\end{align*} 
\item  Prior parameters satisfy $\alpha \in (0, 1)$,  $\tilde{\kappa} \in (0, n)$, $1 \leq   \sqrt{ 1 + \alpha / \gamma} \leq p$,  and  
$$\kappa \geq \max\{ (\alpha + 1 )(4 \din + 6) + t_0 + 3, \; \log \nu \}.$$ 
\item   The weighted adjacency matrices $(B_\sigma^*)_{\sigma \in \bbS^p}$ satisfy 
\begin{equation}\label{eq:beta.min} 
\max_{\sigma \in \bbS^p} \min \left\{  | (B_\sigma^*)_{ij} |^2 \colon  (B_\sigma^*)_{ij} \neq 0   \right\} 
    \geq   35 \kappa  \frac{  \vmax^2 \log p }{\alpha \vmin^2 n}.
\end{equation}  
\end{enumerate}

For informed MH samplers, we will need the following strengthened version of (B3):

\begin{enumerate}[label=(B\arabic*'), start = 3]
\item  Prior parameters satisfy $\alpha \in (0, 1)$,  $\tilde{\kappa} \in (0, n)$, $1 \leq   \sqrt{ 1 + \alpha / \gamma} \leq p$,  and 
$$\kappa \geq  \max\left\{ (\alpha + 1 )(4 \din + 6) + (2 + t_0  )(2 + \xi) + 1, \; \log \nu \right\}$$
for some universal constant $\xi > 0$. 
\end{enumerate}

\subsection{Proofs}\label{sec:proof-sl}

\begin{proof}[Proof of Corollary~\ref{cor:sl_thm1}]
This result is essentially Theorem 6 in~\citet{zhou2021complexity}. 
To prove it, we apply Theorem~\ref{th:rwmh-1}. 
By Section F.4 in~\cite{zhou2021complexity},  we have, for sufficiently large $p$, 
\begin{equation}\label{eq:M-sl}
      \max_{\mathcal{E} \in \mathcal{C}_p(d_\mathrm{in},  d_{\mathrm{out}})}|\cN_\mathrm{ads}(\mathcal{E} )| \leq M = 3 t_0 p^{t_0 + 2} \log_2 p \leq p^{t_0 + 3},
\end{equation} 
and the unimodal condition parameter $R > p^{t_0 + 3}$ with probability $1 - O(p^{-1})$. Hence,  
\begin{align*}
    \tau(\mathsf{P}_0^{\text{lazy}}, \epsilon) = O_p \left(  p^{t_0 + 2} \log p \log \frac{1}{\epsilon \pimin } \right), \text{ where } \pimin =  \min_{\mathcal{E} \in \mathcal{C}_p(d_\mathrm{in}, d_{\mathrm{out}})} \pi_n(\mathcal{E}).  
\end{align*} 
By Corollary~4 in \cite{zhou2021complexity} and Assumptions (B2) and (B3), 
\begin{align}
\log\left( \frac{1}{ \pimin } \right)
& \leq 
p(d_{\mathrm{in}} + d^*) \log \left(  p^{\kappa} \sqrt{1 + \alpha/\gamma} \right)
+ \frac{p(\alpha n + \tilde{\kappa})}{2} \log \left( 2 \nu^2 \right)
+ \log 2\\
& \leq 2 p d_{\mathrm{in}}  \log  p^{\kappa + 1}  
+  np  \log \left( 2 \nu^2 \right)  
+ \log 2  
 =  O( n p \kappa ) ,  \label{eq:pimin-sl}
\end{align} 
where in the last step we have used $\din \log p = o(n)$. 
\end{proof}

\begin{proof}[Proof of Corollary~\ref{cor:sl_thm3}]
By Section F.4 in~\cite{zhou2021complexity} and~\eqref{eq:M-sl},  we have, the unimodal condition holds with high probability with $R = p^{ (2 + t_0  )(2 + \xi) + 1} > M^{2 + \xi}$. 
Hence, we can apply Theorem~\ref{th:mix-imh} with 
\begin{equation}
    \ell = 3 t_0 p^{t_0 + 2} \log_2 p, \quad L = \ell^{2 + \xi}. 
\end{equation}  
which yields the first asserted bound. 

To apply Theorem~\ref{th:mix-imh-drift}, we only need to verify the additional condition~\eqref{eq:drift_pimin_condition}. But this follows from the assumption $n = o(p^{(t_0 + 2 )\xi - 1})$ and~\eqref{eq:pimin-sl}.  Hence, 
\begin{equation}
        \tau(\mathsf{P}_h^{\text{lazy}}, \epsilon) = O_p \left(  \frac{\log 1/ \pimin  }{\log \ell }  \right)    = O_p \left(   \frac{n p  \kappa }{\log p}   \right), 
\end{equation}
which concludes the proof. 
\end{proof}

\section{Simulation for variable selection with highly correlated design}\label{sec:supp_bvs} 
When the design matrix exhibits a high degree of collinearity, we expect that $\pi$ can be highly multimodal, and local modes may occur on $(\cV_s, \cN_1)$ even for small $s$. 
Hence, unlike in the moderately correlated setting where we only need to initialize the sampler at a sufficiently sparse model, we may need to impose much stronger assumptions on the initialization so that the samplers can quickly find the true model. 

To investigate this, we generate 100 replicates of $(\bX, \by)$ under the highly correlated setting described in Section~\ref{subsec:bvs}, where  the covariance matrix $\bm{\Sigma}$ satisfies $\bm{\Sigma}_{jk} = e^{-|j-k|/4}$ for $j \neq k$, and  $\bm{\Sigma}_{jj} = 1$ for $j \in [p]$.  
We consider two initialization schemes studied in~\citet{atchade2021approximate}, and denote them by $\delta^{\mathrm{bad}}$ and $\delta^{\mathrm{good}}$. 
Both $\delta^{\mathrm{bad}}$ and $\delta^{\mathrm{good}}$ include 50 randomly sampled variables that are not in $\delta^*$ (i.e., false positives), and for each $j \in [5]$, $\delta^{\mathrm{good}}_j = 1$ while $\delta^{\mathrm{bad}}_j = 0$. Hence, $\delta^{\mathrm{bad}}$ satisfies $||\delta^{\mathrm{bad}}||_1 = 50$ and has 5 false negatives, and  $\delta^{\mathrm{good}}$ satisfies $||\delta^{\mathrm{good}}||_1 = 55$ with no false negative. 
Table~\ref{table:init_bvs} summarizes the results. When started at $\delta^{\mathrm{bad}}$, both samplers fail to find $\delta^*$ in most replicates, and IMH is still more sensitive to nearby local modes as in the simulation study with moderately correlated design. 
When started at  $\delta^{\mathrm{good}}$, both algorithms recover $\delta^*$ quickly, and IMH has a slightly higher success rate than RWMH.  
This result suggests that $\pi$ is probably highly multimodal among models with false negatives (recall that in Example~\ref{ex:vs3}, the posterior probability can decrease when we add some variable that is in the true model), and the identification of $\delta^*$ requires a warm start that already includes $\delta^*$ as a submodel. 

\begin{table}[!h]
\centering
\caption{
Results for 100 replicates with highly correlated design.  
}
{
\footnotesize

\begin{tabular}{ccccc}
\hline
&\multicolumn{2}{c}{RWMH} & \multicolumn{2}{c}{IMH} \\
Initialization & $\delta^{\mathrm{bad}}$ & $\delta^{\mathrm{good}}$ &$\delta^{\mathrm{bad}}$ & $\delta^{\mathrm{good}}$ \\
\hline

Success   & 27 & 95 & 3 &100\\
$H_\mathrm{true}$ & -- & 2198 & -- & 51\\
Time  & 13.2 & 13.5 & 13.1 &19.6\\
$T_\mathrm{true}$ & -- & 9.5 & -- & 8.7\\
\hline
\end{tabular}
}
\label{table:init_bvs}
\end{table}

\end{appendix}

\bibliographystyle{plainnat}
\bibliography{reference_arxiv}

\end{document}